%% file: BR_ArXiv.tex
\documentclass[aps,prd,twocolumn,showpacs,amsmath,amssymb,amsthm,nofootinbib, preprintnumbers]{revtex4-2}
\usepackage{amsmath,amsthm}
\usepackage{graphicx}
\usepackage{epstopdf}
\usepackage{float}
\usepackage{hyperref}
\usepackage{color}
\usepackage[T1]{fontenc}
\usepackage[utf8]{inputenc}
\usepackage[toc,page]{appendix}
\usepackage[usenames,dvipsnames]{xcolor}
\usepackage[normalem]{ulem}
\usepackage{lipsum, babel}
\usepackage{soul}

\DeclareMathOperator{\arcCoth}{arcCoth} 
\DeclareMathOperator{\arcCot}{arcCot}

\newcommand{\tcr}{\textcolor{red}}
\newcommand{\tcb}{\textcolor{blue}}

\newcommand{\be}{\begin{equation}}
	\newcommand{\ee}{\end{equation}}
\newcommand{\ba}{\begin{eqnarray}}
	\newcommand{\ea}{\end{eqnarray}}

\newcommand{\beq}{\begin{equation}}
	\newcommand{\eeq}{\end{equation}}
\newcommand{\beqa}{\begin{eqnarray}}
	\newcommand{\eeqa}{\end{eqnarray}}

\newtheorem{theorem}{Theorem}

\begin{document}
	
	\title{{Bertotti--Robinson {and Bonnor--Melvin universes} 
    in nonlinear electrodynamics}}

		\author{David Kubiz\v n\'ak}
		\email{david.kubiznak@matfyz.cuni.cz}

        \author{Otakar Sv{\'i}tek}
		\email{otakar.svitek@matfyz.cuni.cz}

        \author{Tayebeh Tahamtan}
		\email{tahamtan@utf.mff.cuni.cz}
        
        \affiliation{Institute of Theoretical Physics, Faculty of Mathematics and Physics,
			Charles University, Prague, V Hole{\v s}ovi{\v c}k{\' a}ch 2, 180 00 Prague 8, Czech Republic}

		\date{April 30, 2026}

		\begin{abstract}
			We review the status of Birkhoff's theorem in the presence of 
            nonlinear electrodynamics (NLE) --  extending the analysis to the case without asymptotic flatness. This leads to the Bertotti--Robinson-type (direct product) geometry with generally unequal radii for its $AdS_{2}$ and $S_{2}$ factors, determined by a given NLE model. As can be expected, such a geometry can also be recovered from a near-horizon limit
            of the corresponding extremal  
            NLE charged black hole (if it exists). These extremal black holes are shown to be linearly stable for specific NLE models, unlike in the Maxwell--$\Lambda$ case where unequal radii also arise in near-horizon geometry. Regular particle-like models are constructed by replacing the interior of these black holes with corresponding Bertotti--Robinson-type geometry. We also revisit the NLE generalization of the Bonnor--Melvin universe, describing a regular axisymmetric configuration of magnetic field lines in gravito-magnetic equilibrium. Explicit examples are derived for the Maxwell, Born--Infeld,  RegMax, and Frolov--Hayward theories of electrodynamics. 
		
		\end{abstract}

		\maketitle


\section{Introduction}

Original Birkhoff's theorem \cite{Birkhoff:1923relativity} (derived earlier by Jebsen \cite{Jebsen:1921,Jebsen:2005hkb}) is formulated for solutions of Einstein equations in the absence of matter and states that any vacuum spherically symmetric solution must be static and asymptotically flat. This naturally leads to the Schwarzschild solution.

A generalization of this theorem to the electro-vacuum case of the Einstein--Maxwell theory
gives Reissner--Nordstr\"{o}m metric as a unique asymptotically flat solution. However, in such a case, there also exists a spherically symmetric but not asymptotically flat solution -- the so called {\em Bertotti--Robinson} universe \cite{Bertotti:1959, Robinson:1959, Dolan:1968, Tariq1974} -- leading to a nontrivial generalization of Birkhoff's theorem. More extensions of this theorem including the cosmological constant and other types of sources are gathered in \cite{Book-Frolov2011}.  More recently, an extension of Birkhoff's theorem for the theory of {\em nonlinear electrodynamics (NLE)} was  discussed in \cite{Diaz:2019-Birkhoff}, concentrating on the asymptotically flat case. The question we want to address here is: what type of geometry one obtains when the condition of asymptotic flatness is relaxed for the  spherically symmetric Einstein-NLE system. Or, in other words, what is the NLE counterpart to the Bertotti--Robinson universe?

The Bertotti--Robinson electromagnetic universe \cite{Bertotti:1959, Robinson:1959, Dolan:1968, Tariq1974} is a conformally flat source-free solution of the Einstein--Maxwell equations with a non-null electromagnetic field.
This spacetime is formed by a direct product of a two-dimensional AdS spacetime ($AdS_{2}$) and a sphere ($S_{2}$). Hairy extensions of the Bertotti--Robinson solution were  presented in  \cite{Dereli_2021, Mazharimousavi_2022}. 

Another type of Einstein--Maxwell solution filled with a non-null magnetic (or electric) field with cylindrical instead of spherical symmetry is the {\em Bonnor--Melvin universe}, representing a bundle of magnetic (electric) field lines in {\em gravito-electromagnetic equilibrium}. This solution was originally discovered by Bonnor \cite{Bonnor:1954tis} and later by Melvin \cite{Melvin:1963qx}. Similarities between the Bonnor--Melvin and Bertotti--Robinson spacetimes are discussed in \cite{Garfinkle_2011}. The NLE generalization of the Bonnor--Melvin solution was studied in \cite{Gibbons_2001} using the Legendre dual description for a nonlinear electromagnetic field. In that paper, the authors have shown how to obtain the solution as a series expansion, or up to an integration, the primary example being the Born--Infeld theory \cite{Born:1934gh}. (See also \cite{Barrientos:2024umq} for a recent discussion of Bonnor--Melvin-type solutions in ModMax theory and \cite{Bandos:2020jsw} for a Bonnor--Melvin-like solution with baryonic charge.)  Here we will reformulate the derivation using the standard description of NLE and provide a couple of examples of Bonnor--Melvin-type spacetimes in various theories of NLE, focusing in particular on the recently studied RegMax electrodynamics \cite{Tahamtan:2020lvq, Kubiznak:2022vft}, which, as we will see and similar to the Maxwell case, yields a regular Bonnor--Melvin spacetime in the full admissible parameter range.

Our paper is organized as follows. 
To set up a stage, 
we gather equations defining theories of NLEs in Sec.~\ref{sec2} and review  the asymptotically flat case and the corresponding NLE version of Birkhoff's theorem in Sec.~\ref{Sec:3}. The generalized Bertotti--Robinson spacetimes are studied in Sec.~\ref{Sec:4}. Their alternative derivation via the near horizon limit of extremal NLE charged black holes is presented in Sec.~\ref{Sec:4.5}, together with comments on linear stability of the extremal black holes and the construction of the corresponding regular particle models. The  generalization of the Bonnor--Melvin universe is studied in Sec.~\ref{Sec:5} using i) the direct integration of the corresponding equations of motion and ii) the double Wick rotation of the corresponding planar
electrically charged solutions. We conclude in Sec.~\ref{Sec:6}. In 
Appendix~\ref{App} we derive spherical black hole solutions in the recently formulated Frolov--Hayward electrodynamics \cite{Frolov:2025ddw}, used in the main text.

\section{Theories of nonlinear electrodynamics}
\label{sec2}
Theories of NLE that are minimally coupled to Einstein's gravity are derived from the following action:
\begin{equation}\label{bulkAct}
	I= \frac{1}{16\pi} \int_{M} d^4 x\,\sqrt{-g}\Bigl({\cal R}+4\,{\cal L({\cal S}, {\cal P})}\Bigr)\,.
\end{equation}
Here, ${\cal L({\cal S}, {\cal P})}$ is the electromagnetic Lagrangian, a function of the two  electromagnetic invariants 
\be
{\cal S}=\frac{1}{2}F_{\mu\nu}F^{\mu\nu}\,,\quad 
{\cal P}=\frac{1}{2}F_{\mu\nu}(*F)^{\mu\nu}\,,
\ee
where, as usual, we have $F_{\mu\nu}=\partial_\mu A_\nu-\partial_\nu A_\mu$, in terms of the vector potential one-form $A_\mu$.  For NLE theory in vacuum,  the \emph{generalized Maxwell equations} read
\be\label{FE}
d*D=0\,,\quad
dF=0\,, 
\ee
where 
\be\label{Edef}
D_{\mu\nu} ={-2 \frac{\partial \mathcal{L}}{\partial F^{\mu\nu}}
=-2\Bigl({\cal L_S}F_{\mu\nu}+{\cal L_P}*\!F_{\mu\nu}\Bigr)\,, }
\ee
with the subscripts on ${\cal L}$ denoting differentiation.
The corresponding electric and magnetic charges can then be calculated from the standard Gaussian type integrals:
\be\label{QeQm} 
Q_e=\frac{1}{4\pi}\int_{S^2} *D\,,
\quad 
Q_m=\frac{1}{4\pi}\int_{S^2} F\,.
\ee 
Additionally, the equations of motion for the coupled system \eqref{bulkAct} are
\be \label{Hmunu}
G_{\mu\nu}=8\pi T_{\mu\nu}\,,
\ee
where the generalized EM energy-momentum tensor is given by
\be\label{Tmunu}
T^{\mu\nu}=-\frac{1}{4\pi}\Bigl(2F^{\mu\sigma}F^{\nu}{}_\sigma {\cal L_S}+{\cal P}{\cal L_P} g^{\mu\nu}-{\cal L}g^{\mu\nu}\Bigr)\,.
\ee

In this work, we consider generic NLE model of ${\cal L(S,P)}$ or ${\cal L(S)}$ type, but the results are illustrated on specific examples. Apart from linear Maxwell theory, we {will consider its maximally symmetric generalization known as the ModMax theory \cite{Bandos:2020jsw, Kosyakov:2020wxv}, as well as the famous Born--Infeld electrodynamics \cite{Born:1934gh} and the recently derived  RegMax  \cite{Tahamtan:2020lvq, Kubiznak:2022vft} {and Frolov--Hayward \cite{Frolov:2025ddw}   models. The latter three theories provide examples of NLEs with finite self-energy of point charges.}

{
The ModMax theory \cite{Bandos:2020jsw, Kosyakov:2020wxv} is characterized by a dimensionless parameter $\gamma\geq0$ and its Lagrangian reads:  
\be \label{NLE:MM}
{\cal L}_{\mbox{\tiny ModMax}}=-\frac{1}{2}\Bigl({\cal S}\cosh \gamma +\sinh\gamma\sqrt{{\cal S}^2+{\cal P}^2}\Bigr)\,.
\ee 
It is the most general theory that shares all the symmetries with the Maxwell electrodynamics, namely the conformal and electromagnetic duality invariance. 
The Maxwell Lagrangian 
\be 
{\cal L}_{\mbox{\tiny M}}=-\frac{1}{2}{\cal S}\,,
\ee 
is recovered upon setting $\gamma=0$.

The Lagrangian of Born--Infeld electrodynamics \cite{Born:1934gh} can be cast in the following form:
\be\label{BI}
{\cal L}_{\mbox{\tiny BI}}=b^2\Bigl(1-\sqrt{1+\frac{\cal S}{b^2}-\frac{{\cal P}^2}{4b^4}}\Bigr)\,, 
\ee
where the parameter $b$ gives the maximum allowed field strength. In geometric units, $[b] =(\mbox{length})^{-1}$. The Maxwell theory is recovered in the limit $b \rightarrow \infty$.

{The RegMax (Regularized Maxwell) model of NLE was derived in} \cite{Tahamtan:2020lvq, Kubiznak:2022vft}  and its properties were discussed in \cite{Hale:2023dpf, Tahamtan:2023tci, Hale:2024lzh}. The RegMax Lagrangian has the following form:
\begin{equation}\label{Tay}
{\cal L}_{\mbox{\tiny RegMax}}=-2\alpha^4\,\Bigl(1-3\ln(1-s)+\frac{s^3+3s^2-4s-2}{2(1-s)}\Bigr)\,,
\end{equation}
where 
\be \label{s}
s=\Bigl(-\frac{\mathcal{S}}{\alpha^4}\Bigr)^\frac{1}{4}\in (0,1)\,.
\ee
The theory is characterized by a dimensionfull parameter $\alpha>0$, $[\alpha^2]=(\mbox{length})^{-1}$, and reduces to the Maxwell theory when $\alpha \rightarrow \infty$. Although the Lagrangian \eqref{Tay} is rather complicated, the point charge field corresponds to the most straightforward regularization of Maxwell theory \cite{Hale:2023dpf} and it provides many standard {\em self-gravitating} 
solutions beyond spherical symmetry \cite{Hale:2023dpf, Tahamtan:2023tci}.

Most recently, a new,  similar to RegMax model of NLE, was proposed by  Frolov, to model Hayward-type regular black holes via the generalized double copy formalism \cite{Frolov:2025ddw}. The corresponding {{\em Frolov--Hayward (FH)} Lagrangian reads 
\be\label{Frolov-HaywardNLE} 
{\cal L}_{\mbox{\tiny FH}}=-{\beta^4}\bigl(\hat s+\ln(1-\hat s)\bigr)\,,\quad \hat s=\sqrt{-{\cal S}/\beta^4}\,,
\ee 
}
where $\beta$ is dimensionfull, with dimensions $[\beta]=(\mbox{length})^{-1}$, and the Maxwell limit recovered upon setting $\beta\to \infty$,  see Appendix~\ref{App} for more details and for a discussion of spherical black holes in this theory.


\section{Asymptotically flat case: generalized Birkhoff's theorem}\label{Sec:3}

In order to analyze the uniqueness of spherically symmetric NLE solutions we will split our investigation into two parts. In this section, we will address the asymptotically flat case using 
{the established results in \cite{Book-Frolov2011} and essentially recovering the generalized Birkoff's theorem in \cite{Diaz:2019-Birkhoff}.} In the next section, we will analyze a generalization of Bertotti--Robinson solution arising for NLE models.




To derive uniqueness for spherically symmetric solutions in the asymptotically flat case, we can use the following theorem derived in \cite{Book-Frolov2011} (page 166):
\begin{theorem}[{\bf Generalized Birkhoff's theorem}]\label{theorem1}{If the energy-momentum tensor generating a spherically symmetric gravitational field of the form
\be\label{ss-metric}
ds^2= \gamma_{AB}\, dx^{A}dx^{B}+R(x)^2\,d\Omega^2\,,
\ee
{where $d\Omega^2=d\theta^2+\sin^2\!\theta d\varphi^2$ is the standard element on a 2-sphere and $R(x)$ is a non-constant function,}
obeys the following condition:
\be 
T_{AB}=\frac{1}{2}\,T^C{}_C \gamma_{AB}\,,
\ee
then the corresponding solution of the Einstein equations possesses an additional Killing vector field.}
\end{theorem}

 In the above equation,  $\gamma_{AB}=\gamma_{AB}(x)$ is a two-dimensional metric with {$x^A=(t,r)$}, 
 and $T_{AB}$ are the corresponding components of the energy momentum tensor. {The assumption that the radial function $R(t,r)$ is non-constant is crucial; the extra Killing vector is } constructed using derivatives of $R(t,r)$ and shown to be timelike.
 Furthermore, 
$R(t,r)$ is constant along this Killing field flow. If we use coordinates adapted to this new symmetry ($t$ now being the coordinate along the stationary Killing vector) this leads to an important conclusion that $R_{,t}=0$ {and $\gamma_{AB,t}=0$. At the same time, it is possible to diagonalize $\gamma_{AB}$, while preserving the above conditions.}


{Let us now consider the following ansatz for the NLE vector {potential:
 \be \label{general-A}
A=\phi(t,r) dt+p \cos\theta d\varphi\,,
 \ee
where $\phi$ is an electric potential, and the constant $p$ represents the magnetic charge.} It follows that the corresponding NLE energy-momentum tensor \eqref{Tmunu} for arbitrary model of ${\cal L}({\cal S}, {\cal P})$ 
obeys the condition of the above theorem, and we immediately have $R_{,t}=0$ for general NLE theory. 
}
Subsequently, one can introduce $R(r)$ as a new radial coordinate (provided $R_{,r}\neq 0$), which can be accomplished simply by using $R(r)=r$ from now on. 
 

Now, having established that we can assume $R=r$ and {that $\gamma_{AB}$  is a time-independent diagonal metric,}  
we can use equation ${G^{t}{}_{t}=G^{r}{}_{r}}$ (since $T^{t}{}_{t}=T^{r}{}_{r}$ is satisfied automatically) to obtain $-\gamma_{tt}\gamma_{rr}=const.$, where we can choose the constant to be one without loss of generality. Similar reasoning was considered for general sources in \cite{Jacobson:2007tj} and for NLE in our previous paper \cite{Kubiznak:2022vft}. 

Therefore, the final metric anzats is of the standard static and spherically symmetric form
\ba \label{metric-SSS}
ds^2=-f\,dt^2+\frac{dr^2}{f}+r^2d\Omega^2\,.
\ea
The Einstein equations then provide the following unique solution for arbitrary NLE source: \be\label{f_sss}
f=1-\frac{2M}{r}+\frac{\ell_{0}\,\phi}{r} +\frac{1}{2r} \int{\left(r^2{\cal L}+2p\,{\cal L}_{\cal P}\,\phi_{,r}\right)dr}\,,
\ee
{where we used the following expression for ${\cal L}_{\cal S}$ coming from the generalized Maxwell equations \eqref{FE}:
\be\label{LS-spherical}
{\cal L}_{\cal S}=\frac{1}{\phi_{,r}\,r^2}\bigl(\ell_0-p\,{\cal L}_{\cal P}\bigr)\,.
\ee
Here, {$\ell_{0}$ is an integration constant {and $\phi$ is no longer time dependent}. Using \eqref{QeQm} to compute the electric charge and applying \eqref{LS-spherical} we obtain $Q_{e}=-2\, \ell_{0}$.

The above 
solution is asymptotically flat provided that the NLE model satisfies the weak-field Maxwell correspondence.

\section{Generalized Bertotti--Robinson universe}\label{Sec:4}

In the previous section, we have shown  that if the metric function $R(r)$ in \eqref{ss-metric} is not constant, we obtain the generalization of Birkhoff's theorem for all NLE theories. {In this section we focus on the case when 
\be 
R=R_0=\mbox{const.}
\ee 
As we shall see, this will lead to a generalization of Bertotti--Robinson spacetimes.
}

\subsection{Bertotti--Robinson Universe in Maxwell theory}

{Let us first review the well-known Bertotti--Robinson solution in 
Maxwell (linear) electrodynamics. It is given by the following metric and vector potential:
\ba \label{Bertotti-Robinson}
ds^2&=&\frac{R_0^2}{y^2}(-dt^2+{dy^2})+R_0^2\,d\Omega^2\,,\nonumber\\
A&=&-\frac{e}{y} dt +p\cos\theta d\varphi\,,
\ea
where $R_0=\sqrt{e^2+p^2}$ {and we have denoted the `radial coordinate' by $y$}.
The metric is} conformally flat and represents a direct product geometry $AdS_{2} \times S_{2}$ (with both radius of the sphere and anti-de Sitter curvature radius equal to $R_{0}$). {As such, it admits a 6-dimensional group of $SL(2,{\mathbf R})\times SU(2)$ isometries.} 
Furthermore, both electromagnetic invariants are constants in this case and read:
\be 
{\cal S}=\frac{p^2-e^2}{R_0^4}\,,\quad {\cal P}=\frac{2ep}{R_0^4}\,,
\ee 
indicating that the electromagnetic field is not null. It is easy to check that the Gaussian integrals \eqref{QeQm} identify $e$ and $p$ as electric and magnetic charges, respectively. These can be taken as two independent parameters/charges characterizing the solution.

In what follows, we shall study possible generalizations of this solution to theories of NLE.

\subsection{Generalized static solutions}

{Let us consider the following ansatz for the spacetime geometry and vector {potential:
\ba\label{Bertotti-metric}
ds^2&=&h(y)\bigl(-dt^2+dy^2\bigr)+R_0^2d\Omega^2\,,\nonumber\\
A&=&\phi(y) dt+p \cos\theta d\varphi\,,
\ea
}
where the parameter $p$ is related to the magnetic charge, and $\phi$ is a scalar (electric) potential. In here, we have used  the fact that all 2-dimensional metrics are conformally flat, as well as focused on the static case with  $h=h(y)$ and $\phi=\phi(y)$.}\footnote{The time dependent case, $h(t,y)$ and $\phi=\phi(t,y)$ can in principle also be considered. W leave this for future studies. 
}

{
The two electromagnetic invariants then read 
\be 
{\cal S}=\frac{p^2}{R_0^4}-\frac{\phi'^2}{h^2}\,,\quad {\cal P}=\frac{2p\phi'}{h R_0^2}\,,
\ee 
where we have denoted $'=\frac{d}{dy}$\,. In what follows, we will assume that for a generic NLE model $\cal L(\cal S, {\cal P})$ both such invariants are (similar to the Maxwell case) constant.}\footnote{In fact, for a generic NLE model this `assumption' is actually a consequence of the corresponding equations of motion. 
}
Such an assumption then yields
\be
{\cal S}={\cal S}_0=\frac{p^2-\psi_0^2}{R^4_0}\ ,\quad 
{\cal P}={\cal P}_0=\frac{2\,\psi_0\,p}{R^4_0}\,,
\ee
where $\psi_0$ is a constant, and we  
have the following formula for the  scalar potential: 
\be\label{potential-Bertotti}
\phi(y)= \frac{\psi_0}{R_0^2}\int h(y) dy\,.
\ee
}


{Since both invariants ${\cal S}$ and ${\cal P}$ are constant, so is the Lagrangian itself, and consequently also the components of the energy momentum tensor $T^\mu{}_\nu$.} Specifically, we have $8\pi T^{\mu}{}_{\nu}=diag(-T_0, -T_0, T_1,	T_1)$ (with $T_0$ and $T_1$ being constants).  Considering the form of electromagnetic potential \eqref{Bertotti-metric}, 
\eqref{potential-Bertotti}  these constants read 
{
\ba\label{T-Bertotti}
T_0&=&-2\,{\cal L} - \frac{4\,\psi_{0}^{2}}{R_{0}^{4}}\,{\cal L}_{\cal S} + 2{\cal P}\,{\cal L}_{\cal P}\,, \nonumber\\
T_1&=&2\,{\cal L} - \frac{4\,p^{2}}{R_{0}^{4}}{\cal L}_{\cal S} - 2\,{\cal P}\,{\cal L}_{\cal P}\,,
\ea
}
In order to find the metric, we have to solve the Einstein field equations. Starting from  $G^{t}{}_{t}={8\pi} T^{t}{}_{t}$, where $G^{t}{}_{t}=G^{y}{}_{y}=-\frac{1}{R^2_0}$, we get 
\be \label{BR-EEtt}
T_0=\frac{1}{R^2_0}\,.
\ee
From the above equation it is clear that $T_0>0$, which is in agreement with the weak energy condition (WEC: $T_0>0$ and ${\cal L}_{\cal S}<0$). Moreover, from \eqref{T-Bertotti}, we always have $T_1>-T_0$ for ${\cal L}_{\cal S}<0$ .

The last remaining Einstein equation is $G^{\theta}{}_{\theta}-{8\pi} T^{\theta}{}_{\theta}=0$, with 
\ba
G^{\theta}{}_{\theta}=G^{\varphi}{}_{\varphi}=\frac{h\,h_{,yy}-h^2_{,y}}{2h^3}\,,
\ea
and $8\pi T^{\theta}{}_{\theta}=T_1$. Introducing $(\log{h})_{,y}$ as a new function, one can reduce the order and integrate the equation to obtain
\be \label{BR-constant invariant}
\pm \int \frac{dh}{2h\,\sqrt{h\,T_1\pm h_1}}=y+h_0\,,
\ee
where  $h_0$ and $h_1>0$ are integration constants. The above integral has three different solutions depending on the sign in the square root and the value of constant $h_1$  
\ba
h=
\begin{cases}\label{f-cases}
\frac{[\tanh \sqrt{h_1}({y+h_0})]^2-1}{T_1/h_1}\,, \quad \quad (+)\\
\frac{[\tan \sqrt{h_1}({y+h_0})]^2+1}{T_1/h_1}\,,\,\,\, \quad \quad (-)\\
\frac{1}{T_1\,(y+h_0)^2}\,,\quad \quad \quad \quad \quad  h_1=0\ .
\end{cases}
\ea
These solutions are insensitive to the overall sign in front of the integral in \eqref{BR-constant invariant}. The first solution is negative for all values of $y$ unless the energy momentum component $T_1$ is negative.\footnote{ Note that the strong energy condition (SEC) leads to $T_1>0$.} Negative $h$ would correspond to the time-dependent (cosmological) metric which we ruled out at the beginning. The second solution would be interesting for certain range of $y$ only. The final one resembles the Bertotti--Robinson solution \eqref{Bertotti-Robinson} if we set $h_0=0$, {which can be achieved without loss of generality by a coordinate transformation that preserves \eqref{Bertotti-metric} and \eqref{potential-Bertotti}. In what follows, we focus on this case.} 

{Thus, for any NLE, we have the following solution:
\ba \label{final-metric-Bertotti}
ds^2&=&\frac{1}{T_1\,y^2}(-dt^2+{dy^2})+\frac{1}{T_0}\,d\Omega^2\,,\nonumber\\
A&=&-{\frac{e}{y}dt}+p\cos\theta d\varphi\,,
\ea
where {we have introduced the electric charge parameter  $e\equiv \psi_0/(R_0^2 T_1)$, in terms of the previous parameters. 
The solution \eqref{final-metric-Bertotti} is characterized by 2 {\em independent} parameters. Namely, it is described by 4 constants, 
\be \label{4-constants}
\{T_0, T_1, e, p\}\,,
\ee
which are subject to two implicit constraints \eqref{T-Bertotti}, which now read 
{
\ba\label{T-Bertotti2}
T_0&=&-2{\cal L}-4e^2T_1^2{\cal L}_{\cal S}+2{\cal P}\,{\cal L}_{\cal P}\,, \nonumber\\
T_1&=&2{\cal L}-4p^2T_0^2{\cal L}_{\cal S}-2{\cal P}\,{\cal L}_{\cal P}\,,
\ea
}
where 
\be 
{\cal S}=p^2T_0^2-e^2T_1^2\,,\quad {\cal P}=2ep T_0 T_1\,.
\ee 
The two electromagnetic parameters $e$ and $p$ are then related to the electric and magnetic charges via the Gaussian integrals \eqref{QeQm}, namely\footnote{
Note the (perhaps puzzling) appearance of parameter $p$ in the electric charge formula \eqref{QeQmBR}; see also the explicit examples below. 
}
\be\label{QeQmBR} 
Q_e=-2e{\cal L}_{\cal S}\frac{T_1}{T_0}+2p{\cal L}_{\cal P}\,,\quad Q_m=p\,.
\ee 

The above spacetime \eqref{final-metric-Bertotti}
has the structure of a Bertotti--Robinson universe, since it is a direct product of the two-sphere and two-dimensional anti-de Sitter spacetime. Namely, $T_0$ determines the radius of $S_2$ according to \eqref{BR-EEtt}, and the explicit relation between $T_1$ and the radius of curvature of $AdS_{2}$  is given by 
$T_1={1}/{\ell_2^2}$. We call it a  {\em generalized Bertotti--Robinson} solution from now on. Whereas for the Maxwell theory, the energy momentum components have the same values: ${T_0=T_1}$, and thence both curvature radii of the direct product factors are the same, for NLE theories these radii are in general different. Moreover, the spacetime is no longer conformally flat unless $T_0=T_1$ (the Weyl tensor is proportional to $T_0-T_1$). 
Based on the results we can state the following observation:
\begin{theorem}\label{theorem2} Conformally invariant NLE leads to equal radii for $AdS_{2}$ and $S_{2}$ and therefore provides conformally flat solution (standard Bertotti--Robinson universe).
\end{theorem}
\begin{proof}
{The Ricci scalar for \eqref{final-metric-Bertotti} and the trace $T$ of energy momentum tensor \eqref{T-Bertotti} satisfy ${\cal R}=-8\pi T=2\left(T_0-T_1\right)$. Conformally invariant NLEs have vanishing trace of the energy momentum tensor, which leads to $T_0-T_1=0$. Since, as already mentioned, the Weyl tensor is also proportional to $T_0-T_1$, it automatically vanishes and we end up with conformally flat geometry.}
\end{proof}
In particular, this means that the ModMax theory leads to a conformally flat Bertotti--Robinson spacetime, though with different values of radii when compared to the Maxwell case. 

We illustrate the above results for couple of NLE models 
in Table~\ref{Table:BR}. Here, $T_0$ and $T_1$ are calculated from \eqref{T-Bertotti2},
and the charge $Q_e$ is determined from \eqref{QeQmBR}. For the RegMax and Frolov--Hayward electrodynamics the full analytic solution of  \eqref{T-Bertotti2} is not possible; in these cases we provide the leading order correction to Maxwell, expanding in the due NLE parameter.

Note also that for the Born--Infeld case, we have the following bound on the existence of generalized Bertotti--Robinson solution: \be
4b^2(Q_e^2+Q_m^2)> 1\,,
\ee
coming from positivity of $T_0$. As we shall see in the next section, this bound corresponds to the non-existence of weakly charged extremal black holes in this theory.






\renewcommand{\arraystretch}{1.5}

\begin{center}
\setlength{\tabcolsep}{6pt}
\begin{table*}[t!]

\begin{tabular}{ |c||c|c|c|}
 \hline
   \mbox{NLE theory}& $T_0$ & $T_1$& $Q_e$ \\
 \hline
 \hline
 \mbox{Maxwell} & $T_0^{\mbox{\tiny M}}=\frac{1}{e^2+p^2}$   &$T_1^{\mbox{\tiny M}}=\frac{1}{e^2+p^2}$&   $\frac{T_1 e}{T_0}=e$ \\
 \hline
 \mbox{ModMax} & $\frac{\exp(\gamma)}{e^2+p^2 \exp(2\gamma)}$   &$\frac{\exp(\gamma)}{e^2+p^2 \exp(2\gamma)}$&   $\frac{T_1 e}{T_0} \exp(-\gamma)=e \exp(-\gamma)$ \\
 \hline
 \mbox{Born--Infeld}  & $\frac{4b^2}{4b^2(e^2+p^2)-1}$   & $\frac{4b^2}{4b^2(e^2+p^2)+1}$ &  $\frac{T_1e}{T_0}\sqrt{\frac{b^2+p^2T_0^2}{b^2-e^2T_1^2}}=e$  \\\hline
  \mbox{RegMax} &  $T_0^{\mbox{\tiny M}}-\frac{4(e^2-p^2)^{1/4}(e^4+7e^2p^2+2p^4)}{5\alpha(e^2+p^2)^{7/2}}+\dots $  & $T_1^{\mbox{\tiny M}}-\frac{4(e^2-p^2)^{\frac{1}{4}}(2e^4+7e^2p^2+p^4)}{5\alpha(e^2+p^2)^{\frac{7}{2}}}+\dots$&  $\frac{T_1e}{T_0}\frac{\alpha}{(\alpha-(T_1^2e^2-T_0^2p^2)^{1/4})^2}=e$\\
  &&&$+\frac{2e(3e^2+7p^2)(e^2-p^2)^{\frac{1}{4}}}{5\alpha(e^2+p^2)^{\frac{3}{2}}}+\dots$ \\
 \hline
\mbox{Frolov--Hayward} &  $T_0^{\mbox{\tiny M}}-\frac{2p^2\sqrt{e^2-p^2}(5e^2+p^2)}{3\beta^2(e^2+p^2)^{4}}+\dots $  & $T_1^{\mbox{\tiny M}}-\frac{2e^2\sqrt{e^2-p^2}(e^2+5p^2)}{3\beta^2(e^2+p^2)^4}+\dots$&  $\frac{T_1e}{T_0}\frac{\beta^2}{\beta^2-\sqrt{T_1^2e^2-T_0^2p^2}}=e$\\
  &&&$+\frac{e(e^2+5p^2)\sqrt{e^2-p^2}}{3\beta^2(e^2+p^2)^2}+\dots$ \\
 \hline
 \mbox{Maxwell-$\Lambda$} &  $\Lambda+\frac{1 {-4\Lambda p^2} + \sqrt{(1+4\Lambda e^2)(1-4\Lambda p^2)}}{2(e^2+p^2)} $  & $\Lambda+\frac{1 {-4\Lambda p^2} + \sqrt{(1+4\Lambda e^2)(1-4\Lambda p^2)}}{2(e^2+p^2)}$&  $\frac{T_1-\Lambda}{T_0+\Lambda}e=e-2e(e^2+p^2)\Lambda$\\
  &&&$\qquad +O(\Lambda^2)$ \\
 \hline
  \mbox{Born--Infeld-$\Lambda$} &  $\frac{2b^2(b^2+2\Lambda b^2(e^2-p^2)+\Lambda^2p^2+B)}{2b^2(e^2+p^2)(2b^2-\Lambda)-B} $  & $\frac{2b^2(b^2+2\Lambda b^2(e^2-p^2)+\Lambda^2p^2+B)}{b^2+4b^4(e^2+p^2)-4\Lambda b^2p^2+\Lambda^2p^2}$&  $\frac{T_1-\Lambda}{T_0+\Lambda}e\sqrt{\frac{b^2+p^2(T_0+\Lambda)^2}{b^2-e^2(T_1-\Lambda)^2}}$\\
  \hline
\end{tabular}
\caption{
{\bf Examples of NLE-BR spacetimes.} In this table we gather a couple of examples of Bertotti--Robinson spacetimes for concrete NLE models, namely, for the Maxwell, ModMax, Born--Infeld, RegMax, and Frolov--Hayward cases. Note that for the latter two Eq.~\ref{T-Bertotti2} cannot be analytically solved and we provide only the leading order correction to Maxwell, expanding in the due NLE parameter. Note also that in these cases the electric charge $Q_e$ depends on the parameter $p$.  In the last two rows we have also included the $\Lambda$ case for the Maxwell and Born--Infeld theories, where for the latter we have introduced the following shorthand: $B\equiv\sqrt{(b^2+4\Lambda b^2e^2-\Lambda^2e^2)(b^2-4\Lambda b^2p^2+\Lambda^2p^2)}$. Obviously, for small enough $\Lambda<0$ such a quantity is positive. 
}
\label{Table:BR}
\end{table*}
\end{center}

\subsection{Inclusion of the cosmological constant}
Let us also briefly comment on the extension of the above results when we consider the cosmological constant. 

\subsubsection{$\Lambda$  only}

First, we will consider the cosmological constant without any NLE. In that case, the energy momentum tensor (when treating the cosmological constant as a source) is  $8\pi\, T^{\mu}{}_{\nu}=diag(-\Lambda, -\Lambda, -\Lambda, -\Lambda)$.
{Employing again the metric ansatz \eqref{Bertotti-metric}, 
from}  $G^{t}{}_{t}=-\Lambda$ we get 
\be 
\frac{1}{R_0^2}=\Lambda\,,
\ee
which implies $\Lambda>0$.} 
From $G^{\theta}{}_{\theta}+\Lambda=0$ we then  obtain  
\be \label{BR-cosmo}
\pm \int \frac{dh}{2h\,\sqrt{ h_1-h\,\Lambda}}=y+h_0\,.
\ee

Considering $h>0$, we necessarily need $h_1\geq \max(h)\,\Lambda>0$,  and thus only have the following solution:
\be
h=-\frac{[\tanh \sqrt{h_1}({y+h_0})]^2-1}{\Lambda/h_1}\ .
\ee 
Computing the maximum of this $h$ plus the inequality for $h_1$ above we arrive at condition $h_1\geq 1$.

If we want to consider other possibilities as in \eqref{f-cases}, we need to assume $h<0$. In our previously favored case $h_1=0$ this leads to (considering $h_0=0$)
\be\label{f-Nariai}
h=-\frac{1}{\Lambda\,y^2}.
\ee
Now $y$ is a time coordinate and the solution corresponds to a two-dimensional de Sitter space in $t,y$ coordinates, while the complete spacetime (setting 
$\Lambda\equiv 3/{\ell^2}$) 
\be\label{Narii}
ds^2=\frac{\ell^2}{3 y^2}\bigl(-dy^2+dt^2\bigr)+\frac{\ell^2}{3}d\Omega^2\,,
\ee 
is the {\em Nariai solution} of the $dS_{2}\times S_{2}$ form \cite{nariai1951new}.

\subsubsection{NLE with $\Lambda$}

If we want to add a cosmological constant to the NLE source analyzed in the previous subsection, one can achieve this via the following mapping:
\be
T_0 \to T_0 + \Lambda \ ,\quad T_1 \to T_1 - \Lambda\ .
\ee
Explicitly, from  $G^{t}{}_{t}=-\Lambda-T_0$ we get 
\be 
\frac{1}{R_0^2}=\Lambda+T_0\,.
\ee
Therefore, one needs to satisfy $\Lambda+T_0>0$ (so we can also have  $\Lambda <0$). And from $G^{\theta}{}_{\theta}+\Lambda-T_1=0$ we obtain
\ba
h=
\begin{cases}\label{fLambda-cases}
\frac{[\tanh \sqrt{h_1}({y+h_0})]^2-1}{(T_1-\Lambda)/h_1}\,, \quad \quad (+)\\
\frac{[\tan \sqrt{h_1}({y+h_0})]^2+1}{(T_1-\Lambda)/h_1}\,,\,\,\, \quad \quad (-)\\
\frac{1}{(T_1-\Lambda)(y+h_0)^2}\,,\quad \quad \quad \ h_1=0\ .
\end{cases}
\ea
All these solutions are possible only when $T_1-\Lambda>0$,  if $h>0$ is desired. Analyzing all possible cases for values of $T_0, T_1$ and $\Lambda$ is beyond the scope of our objectives here. {Focusing again on the $h_1=0=h_0$ case, we get the following solution:
\ba \label{final-metric-Bertotti_Lambda}
ds^2&=&\frac{1}{(T_1-\Lambda)y^2}(-dt^2+{dy^2})+\frac{1}{T_0+\Lambda}\,d\Omega^2\,,\nonumber\\
A&=&-\frac{e}{y}dt+p\cos\theta d\varphi\,.
\ea
Note that inclusion of the cosmological constant would make radii of $AdS_{2}/dS_2$ and $S_{2}$ different even for conformal theories, e.g. Maxwell. 
Obviously, Theorem \ref{theorem2} no longer generalizes to the case with cosmological constant; the presence of the cosmological scale  breaks the conformal symmetry.

In the above, $T_0$ and $T_1$ are given by $\Lambda$ generalization of \eqref{T-Bertotti2}, namely 
\ba\label{T-Bertotti2Lambda}
T_0&=&-2{\cal L}-4e^2(T_1-\Lambda)^2{\cal L}_{\cal S}+2{\cal P}\,{\cal L}_{\cal P}\,, \nonumber\\
T_1&=&2{\cal L}-4p^2(T_0+\Lambda)^2{\cal L}_{\cal S}-2{\cal P}\,{\cal L}_{\cal P}\,,
\ea
where now 
\ba
{\cal S}&=&p^2(T_0+\Lambda)^2-e^2(T_1-\Lambda)^2\,,\nonumber\\
{\cal P}&=&2ep (T_0+\Lambda)( T_1-\Lambda)\,.
\ea 
Similarly, the charges now read
\be\label{QeQmBRLambda} 
Q_e=-2e{\cal L}_{\cal S}\frac{T_1-\Lambda}{T_0+\Lambda}+2p{\cal L}_{\cal P}\,,\quad Q_m=p\,.
\ee 

\subsubsection{Maxwell example}

For example, for the Einstein--Maxwell-$\Lambda$ theory, we get 
\be\label{T0T1MLambda}
T_0=T_1=\Lambda+\frac{1 {-4\Lambda p^2} \pm \sqrt{(1-4\Lambda p^2)(1+4e^2\Lambda)}}{2(e^2+p^2)}\,,
\ee
and the charge reads 
\be \label{e-chargeMLambda}
Q_e=\frac{T_1-\Lambda}{T_0+\Lambda}e\,.
\ee 
The plus roots then recover the Maxwell case in the limit $\Lambda\to 0$.

Let us consider this root and set for simplicity $p=0$. Then we can invert Eq.~\ref{e-chargeMLambda}, obtaining 
\be 
e=\pm \frac{Q_e}{\sqrt{1-4\Lambda Q_e^2}}\,.
\ee 
Plugging this to expressions \eqref{T0T1MLambda}, we obtain the following formulae:
\ba 
T_0+\Lambda&=&\frac{1+\sqrt{1-4\Lambda Q_e^2}}{2Q_e^2}\,,\nonumber\\
T_1-\Lambda&=&\frac{\sqrt{1-4\Lambda Q_e^2}+1-4\Lambda Q_e^2}{2Q_e^2}\,.
\ea 
Obviously, when $\Lambda<0$ (anti de Sitter case), the above generalized Bertotti--Robinson solution always exists. On the other hand, when $\Lambda>0$ (de Sitter case), we get the following restriction:
\be 
|Q_e|\leq \frac{1}{2\sqrt{\Lambda}}\,.
\ee 
This is the Nariai bound.

Interestingly, Eqs.~\eqref{T-Bertotti2Lambda} can also be analytically solved for the Born--Infeld-$\Lambda$ theory -- we provide the corresponding solution in Table~\ref{Table:BR}, leaving its detailed analysis for future studies.

\section{Extremal black holes and near-horizon geometry}\label{Sec:4.5}

\subsection{Near-horizon limit}
It is a well known fact that the near-horizon limit  
of an extremal Reissner--Nordstr\"{o}m solution leads to the Bertotti--Robinson universe \eqref{Bertotti-Robinson}. It is therefore reasonable to expect that similarly a near horizon limit of an extremal NLE charged black hole (if it exists) should lead to the generalized Bertotti--Robinson solution derived above.

To see that  this is indeed the case,
let us assume the existence of an extremal black hole of the form \eqref{metric-SSS}. Such a hole is characterized by a double root of the metric function $f$, \eqref{f_sss}, on the horizon located at $r=R_0>0$. Thus, $f$ has the following expansion: 
\be \label{NearHorizon-f}
f=a_1\,(r-R_0)^2 + O\left((r-R_0)^3\right)\,,
\ee 
where $a_1=\frac{1}{2}f''(R_0)$ is a constant. Taking into account only the 
leading order terms in metric coefficients and introducing new coordinate $u=r-R_0$, we arrive at the following metric: 
\be
ds^2=-{a_{1}\,u^2}dt^2+\frac{1}{a_{1}\,u^2}{du^2}+R_{0}^{2}\,d\Omega^2\,.
\ee
This is a standard form of near-horizon geometry of an extremal black hole that can  easily be transformed into the generalized Bertotti--Robinson form by introducing $y=\frac{1}{a_{1}\,u}$, thus obtaining
\be
ds^2=\frac{1}{a_{1}\,y^2}(-dt^2+{dy^2})+R_{0}^{2}\,d\Omega^2\,.
\ee
Comparing this with \eqref{final-metric-Bertotti} then leads to 
\be\label{identification}
T_0=\frac{1}{R_{0}^{2}}\ ,\quad T_1=a_{1}\ .
\ee
These two quantities can be easily derived for any NLE. Although $T_0$ remains always `the same', $T_1$ is equal to the quadratic coefficient of expansion of $f$. Note that the requirements
\be 
R_0>0\,,\quad a_1>0\,,
\ee 
may impose interesting restrictions on the existence of extremal black holes in a given 
NLE theory, as recently observed in \cite{Hale:2025urg} for the case of Born--Infeld theory.

Note also, that by expanding the scalar potential $\phi(r)$ as $\phi(r)=\phi(R_0)+\phi'(R_0)(r-R_0)+O((r-R_0)^2)$, and using the above change of variables, the (static) vector potential \eqref{general-A} now reads
\be 
A=\frac{\phi'(R_0)}{T_1 y}dt+{p\cos\theta d\varphi}\,,
\ee 
{where we have dropped the trivial gauge term. This  
coincides with the vector potential in \eqref{final-metric-Bertotti}, upon identifying   
\be \label{phi'}
\phi'(R_0)=-eT_1\,.
\ee 
This identification maps the bulk black hole charge $Q$ to the Bertotti--Robinson charge $Q_e$ and provides further consistency check for the near horizon limit.  


\subsection{Examples}

Let us provide a couple of examples for specific NLE theories, to illustrate how the above restriction may concretely appear. For simplicity, in all examples we switch off the magnetic charge, $p=0$, and assume that the bulk charge $Q$ is positive.

\subsubsection{Maxwell}
To warm up, let us start with the Maxwell theory. The corresponding spherically symmetric metric function $f$ {and potential $\phi$ are}
\be 
f=1-\frac{2M}{r}+\frac{Q^2}{r^2}\,,\quad \phi=-\frac{Q}{r}\,.
\ee 
Demanding double root at $r=R_0$, imposes two restrictions on the   parameters $\{Q, M, R_0\}$. This yields 
\be \label{NearHorizon-Maxwell}
R_0=Q\,,\quad a_1=\frac{1}{Q^2}\,.
\ee 
We thus recovered $T_1=T_0=1/{Q^2}$. 
At the same time, one has 
$\phi'(R_0)=Q/R_0^2=1/Q$. In order this agrees with \eqref{phi'}, we have to identify 
\be\label{QeQ} 
Q_e=e=-Q\,,
\ee 
where the first equality follows from Table~\ref{Table:BR}.\footnote{The fact that the bulk $Q$ and the Bertotti--Robinson $Q_e$ charges have opposite signs is related to the opposite orientation of the surface elements in the due Gaussian integrals.} 

Obviously, there is no 
restriction on the existence of extremal black holes in this case -- extremal black holes will exist for any value of the black hole charge $Q>0$.



\subsubsection{Born--Infeld}
For the  Born--Infeld theory, the static metric function {$f$ and the potential $\phi$ take the following form:
\ba \label{SSS-BI}
f&=&1-\frac{2M}{r}+\frac{2b^2}{3}r^2-\frac{2b^2}{r}\int \sqrt{r^4+\frac{Q^2}{b^2}}dr\,,\nonumber\\
\phi&=&-\int \frac{Q\,dr}{\sqrt{r^4+Q^2/b^2}}\,.
\ea
}
The vanishing of $f(R_0)$ and $f'(R_0)$ imposes two constraints on the parameters $\{M, Q, R_0, b\}$. For fixed $Q$ and $b$, these yield
\be \label{formulasBI}
a_1=\frac{4b^2}{4Q^2b^2+1}\,,\quad R_0=\frac{\sqrt{4Q^2b^2-1}}{2b}\,.
\ee 
At the same time, we find 
\be 
\phi'(R_0)=\frac{Q}{\sqrt{R_0^4+Q^2/b^2}}=\frac{4Qb^2}{4Q^2b^2+1}\,.
\ee 
Using \eqref{phi'} and \eqref{identification} we thus identify the relation between the bulk and Bertotti--Robinson charges as in \eqref{QeQ}.

Obviously, $R_0$ only exists when 
\be 
Q>Q_\star=\frac{1}{2b}\,.
\ee 
This is the bound studied in \cite{Hale:2025urg}. This indicates that there exists a charge gap for the existence of extremal Born--Infeld black holes. Namely, weakly charged Born--Infeld black holes do not have an inner Cauchy horizon and therefore cannot be extremal. 
Note that in the Maxwell limit, $b\to\infty$, this bound vanishes, and we recover $R_0=Q$ and $T_1=T_0=1/Q^2$,  as expected.

\subsubsection{RegMax}
For the RegMax theory 
the metric function takes the following form \cite{Hale:2023dpf}:
\ba \label{SSS-RegMax}
f&=&1-2\alpha^2Q+\frac{4\alpha Q^{3/2}}{3r}-\frac{2M}{r}+4r\alpha^3\sqrt{Q}\nonumber\\
&&-4\alpha^4r^2\log\Bigl(1+\frac{\sqrt{Q}}{r\alpha}\Bigr)\,.
\ea 
The vanishing of $f(R_0)$ and $f'(R_0)$ again imposes two constraints on the parameters $\{M,Q,R_0, \alpha\}$. However, in this case, the presence of logarithmic terms does not allow one to analytically solve for $M$ and $R_0$ to produce formulas analogous to \eqref{formulasBI}. However, one can express the logarithmic term from one of the equations. This yields the following formula for the coefficient $a_1$: 
\be 
a_1=\frac{2\alpha^2Q^2-(R_0\alpha+\sqrt{Q})^2}{R^2_0(R_0\alpha+\sqrt{Q})^2}\,.
\ee 
Here, it is to be understood that $R_0=R_0(\alpha,Q)$. Interestingly, even without the analytical knowledge of this expression, we can proceed as follows.  
Demanding that $a_1$ is positive requires positive nominator. In particular, for $R_0\to 0$, this yields 
\be 
Q>Q_\star=\frac{1}{2\alpha^2}\,.
\ee 
This is precisely the charge gap for the existence of RegMax extremal black holes identified first in \cite{Hale:2023dpf} and re-derived here using the
near horizon limit.


\subsubsection{Frolov--Hayward electrodynamics}
{Let us finally turn to the black holes in the recently discovered  Frolov--Hayward electrodynamics. We derive the corresponding spherical solution in Appendix~\ref{App} -- the metric function reads (assuming $Q>0$): 
\ba\label{f-Frolov}
f&=&1-\frac{2\beta^2Q}{3}-\frac{2M}{r}+\frac{4\beta Q^{3/2}}{3r}\arcCot\Bigl(\frac{r\beta}{\sqrt{Q}}\Bigr)\nonumber\\
&&
+\frac{2\beta^4 r^2}{3}\log\left(1+\frac{Q}{\beta^2\,r^2}\right)\,. 
\ea
This then yields
\be 
a_1=\frac{\beta^2(2Q^2-R_0^2)-Q}{R^2_0\,(R_0^2\beta^2+Q)}\,.
\ee 
Demanding that $a_1$ is positive requires positive nominator. In particular, for $R_0\to 0$, this yields 
\be 
Q>Q_\star=\frac{1}{2\beta^2}\,,
\ee 
which agrees with the bound \eqref{Frolov_bound_App} obtained in Appendix~\ref{App} for the existence of Frolov--Hayward extremal black holes.

\subsection{Inclusion of cosmological constant}
The cosmological constant can be trivially added to all spherically symmetric NLE solutions \eqref{metric-SSS}. However, positive $\Lambda$ brings an additional (cosmological) horizon  which forces one to specify which horizons coincide to form an extremal solution. If the inner and outer black hole horizons coincide, we obtain the solution corresponding to the generalized Bertotti--Robinson spacetime. If the outer and cosmological horizons coincide (resulting in negative $f$ near the horizon), we have the Nariai case. {In both cases, the solution can be cast in 
the form \eqref{final-metric-Bertotti_Lambda}, where the curvature radii of $AdS_{2}/dS_2$ and $S_{2}$ are no longer equal and read:}
\be \label{NHLLambda}
\frac{1}{R_0^2}=\Lambda+T_0\,,\quad    a_1=T_1-\Lambda\,.
\ee

{To see this explicitly, let us consider the charged AdS black hole in Maxwell electrodynamics. Such a case was studied in detail in \cite{Romans:1991nq}, and we review it here only for illustration. Namely, in this case the metric function and the potential simply read  
\be 
f=1-\frac{2M}{r}+\frac{Q^2}{r^2}-\frac{\Lambda}{3} r^2\,,\quad 
\phi=-\frac{Q}{r}\,.
\ee 
Demanding double root at $R_0$, yields
\ba 
R_0&=&\sqrt{\frac{1-\sqrt{1-4\Lambda Q^2}}{2\Lambda}}\,,\nonumber\\
a_1&=&\frac{2\Lambda(4Q^2\Lambda-1+\sqrt{1-4Q^2\Lambda})}{(1-\sqrt{1-4Q^2\Lambda})^2}\,.
\ea 
Moreover, employing the identification \eqref{NHLLambda} then yields
\be 
T_0=T_1=-\Lambda+\frac{1+\sqrt{1-4Q^2\Lambda}}{2Q^2}\,,
\ee 
which agrees with \eqref{T0T1MLambda}, upon employing 
\eqref{e-chargeMLambda},  considering $Q_{e}=-Q$, and setting $p=0$.

\subsection{Linear stability of extremal black holes}
Inspired by \cite{Maciej2023},  let us briefly study  linear stability of extremal black holes obtained in previous sections employing a massless scalar field $\Phi$.
Using the spherical symmetry of the background spacetime 
\eqref{metric-SSS}, we expand $\Phi$ into the spherical harmonics,
\be
\Phi=\sum_{\ell, m}\Phi_{\ell m}Y_{\ell m}\,,
\ee
upon which the corresponding wave equation reduces to
\be\label{Klein-gordon}
f\,\Phi^{''}_{\ell m}+\frac{\left(2\,f+r\,f^{'}\right)}{r}\Phi^{'}_{\ell m}-\frac{\ell(\ell+1)}{r^2}\Phi_{\ell m}=0\,.
\ee
Since we are interested in the near-horizon region, as discussed before, the metric function can be approximated by its near-horizon expansion, $f=a_1\,(r-R_0)^2+\dots$, where $a_1>0$. This turns the above wave equation \eqref{Klein-gordon} into an Euler type equation and we can seek the scalar field in the form $\Phi \approx (r-R_0)^{\gamma_{\pm}}$, giving 
\be
\gamma_{\pm}=-\frac{1}{2}\left(1\pm \sqrt{1+\frac{4\ell(\ell+1)}{a_1\,R^2_0}}\right)\,.
\ee

Obviously, $\gamma_{+}$ leads to a singular behavior so we focus on the  $\gamma_{-}$ branch, examining various solutions discussed in the previous sections.

For Maxwell theory (from \eqref{NearHorizon-Maxwell}), we have $a_1\,R^2_0=1$ and $\gamma_{-}=\ell$,  always a positive integer for any $\ell\geq 1$ as is shown in \cite{Maciej2023}.  In general, if we want to satisfy $\gamma_{-}\geq 1$ (which avoids instability as discussed at the end of this section),  we need 
\be\label{l-condition}
a_1 R^2_0 \leq \frac{\ell(\ell+1)}{2}\,.
\ee
For the case of Maxwell plus the cosmological constant we have 
\be 
a_1\,R^2_0=\frac{T_0-\Lambda}{T_0+\Lambda}\,,
\ee
(using that in Maxwell theory we have $T_0=T_1$) and the condition \eqref{l-condition} becomes 
\be 
T_0\left(2-\ell(\ell+1)\right) \leq \left(2+\ell(\ell+1)\right) \Lambda\,.
\ee
For positive $\Lambda>0$ the above relation always holds. But for $\Lambda<0$ it is violated for $\ell=1$ and for $\ell>1$ it incurs a limit on the possible value of $\Lambda$.  Specifically, considering $\ell>1$ and substituting for $T_0$ from \eqref{T0T1MLambda} with $p=0$ we have
\be 
\left(\frac{4\ell(\ell+1)}{\ell(\ell+1)-2}\right)|\Lambda|\leq 1+\sqrt{1-4e^2|\Lambda|}\,.
\ee
This leads to $|\Lambda| \leq \frac{2}{j^2}(j-2e^2)$ where $j=\left(\frac{4\ell(\ell+1)}{\ell(\ell+1)-2}\right)$.

Now let us apply the above analysis to charged black holes in specific NLE models discussed in previous sections. 
For the Born--Infeld model (\eqref{formulasBI}), the expression $a_1 R^2_0=\frac{4Q^2b^2-1}{4Q^2b^2+1}$ is always less than one (note that $Q>\frac{1}{2b}$) which means $\gamma_{-}>1$ for any $\ell\geq 1$. In the Maxwell limit $b\to\infty$ one obtains $\gamma_{-}=1$ for $\ell = 1$ saturating the bound.

For two remaining models, namely RegMax and Frolov--Hayward, the analysis is complicated by the fact that we do not have explicit formulas for $R_0$. But if the following inequality is satisfied for RegMax:
\be\label{RegMax-cond}
\frac{Q^2}{\left(R_0+\sqrt{Q}/\alpha\right)^2} \leq \frac{\ell(\ell+1)+2}{4}\,,
\ee
and the next one for Frolov--Hayward model:
\be\label{Frolov-Hayward-cond}
\frac{2Q^2-R^2_0-Q/\beta^2}{R^2_0+Q/\beta^2} \leq \frac{\ell(\ell+1)+2}{4}\,,
\ee
we again have $\gamma_{-}\geq 1$ for any $\ell\geq 1$. The condition \eqref{RegMax-cond} is strictest for $l=1$ and is exactly satisfied in Maxwell limit $\alpha\to\infty$ since then $R_0=Q$. Analyzing the exact dependence $R_0(\alpha,Q)$ in RegMax model one can see that for finite $\alpha$ we have $R_0+\sqrt{Q}/\alpha>Q$, thus satisfying the inequality. Similar behavior can be shown for Frolov--Hayward model relation \eqref{Frolov-Hayward-cond}, the inequality is sharply satisfied for finite $\beta$ and becomes equality for Maxwell limit $\beta\to\infty$.

Why have we limited ourselves to $\gamma_{-}\geq 1$? If $\gamma_{-}<1$ we have diverging energy momentum tensor of the scalar field (since it contains first derivatives of scalar field) at the horizon indicating instability (namely blow up of certain components of the Ricci tensor). Unlike in the case of extremal Reissner--Nordstr\"{o}m-AdS spacetime \cite{Maciej2023} both Born--Infeld, RegMax and Frolov--Hayward extremal black holes do not suffer from such an instability, although they all share the property of having unequal radii for their near-horizon limit (generalized Bertotti--Robinson solution). 

Although the analysis is only at the linear level of scalar perturbations, it was shown in \cite{Maciej2023} that the above-described potential instability propagates into the full Einstein--Maxwell-AdS picture.

\subsection{Constructing regular particle models}
{The above constructed 
Bertotti--Robinson spacetimes can be used for the construction of regular particle models.}
To achieve this, we glue
two spacetimes in such a way that we replace the black hole region with a completely regular solution, ensuring that no charged matter shell gets induced on the gluing surface -- in other words, ensuring sufficient metric smoothness. Such a particle model for the Einstein--Maxwell theory is provided by gluing the Bertotti--Robinson solution and the extremal Reissner--Nordstr\"{o}m  metric, with a gluing surface identified with the extremal horizon, as described in \cite{ParticleModelZaslavskii}. Later, the same method was applied for the specific NLE model called the Hoffmann--Born--Infeld model in four  \cite{ParticleModelHabib} and three  \cite{ParticleModelTay} dimensions.  Here, we will show that it is possible to glue the due Bertotti--Robinson solution to any NLE extremal black hole. 

To do this, we assume the static spherically symmetric black hole metric \eqref{metric-SSS}, 
\be 
ds^2=-f\,dt^2+\frac{dr^2}{f}+r^2d\Omega^2\,,
\ee 
as an external portion of the spacetime, valid until $r=R_0$, and gluing another spacetime along this surface that will form the inner geometry, using the Israel junction conditions. The metric functions corresponding to inner and outer spacetimes should be continuous at $r=R_0$. Additionally, if the first derivatives are continuous at 
$r=R_0$, we avoid producing any physical sources on the shell. The exact conditions for proper matching can be seen by computing the potentially induced surface energy momentum tensor, $S^{\mu}_{\nu}$ at $r=R_0$. The tensor $S^{\mu}_{\nu}$ is the so-called Lanczos tensor and  can be expressed in terms of the extrinsic curvature tensor \cite{Israel:1966rt, ParticleModelZaslavskii}, 
\be
8\pi\,S^{\mu}_{\nu}=[K^{\mu}_{\nu}]-\delta^{\mu}_{\nu} [K]\,,
\ee
where $K^{\mu}_{\nu}$ is the extrinsic curvature, $K=K^{\mu}_{\mu}$ and $[\ ]=(\ )_+-(\ )_-$ is the jump at the gluing surface. If $[K^{\mu}_{\nu}]=0$, then both regions match smoothly enough and $S^{\mu}_{\nu}=0$.

In the previous section, we showed that all NLE models that give a degenerate horizon provide a metric function of the form $f=a_1(r-R_0)^2 + O((r-R_0)^3)$ \eqref{NearHorizon-f} close to the horizon at $R_0$, where the constants $a_1$ and $R_0$ are specific to a given NLE model and the values of the relevant parameters. Keeping only the dominant term in $f$ we were able to transform the metric into the form of generalized Bertotti--Robinson geometry \eqref{final-metric-Bertotti} characterized by \eqref{identification}. Since the extrinsic curvature for the standard spherically symmetric metric depends only on $f$ and $f'$ it is obvious that for gluing surface identical with the degenerate horizon $r_0=R_0$ of NLE charged black hole we can glue inner generalized Bertotti--Robinson solution with corresponding parameters without introducing any surface matter. Specifically, the evaluation on the horizon gives {$f(R_0)=0=f'(R_0)$ on both sides.} 
One can also easily confirm that the electromagnetic fields from both sides match properly -- thus giving rise to a regular particle model in a given NLE theory.\footnote{{Note that the spacetime is not fully regular, as jumps in higher derivatives across the gluing surface are present; for this reason, we call these models regular particles, rather than regular (extremal) black holes.}  }

\section{ Bonnor--Melvin universe in NLE theories}\label{Sec:5}

Let us turn our attention to the closely related solution that is also filled with a magnetic (electric) field albeit no longer constant -- the Bonnor--Melvin solution \cite{Bonnor:1954tis, Melvin:1963qx, Tahamtan-Melvin}. It represents a bundle of magnetic (electric) field lines held in equilibrium by their electromagnetic repulsion and gravitational attraction. In \cite{Gibbons_2001}, the generalized Bonnor--Melvin universe solution in nonlinear electrodynamics was derived. Here, we shall reformulate the approach to this problem in order to provide more explicit results for NLE models of interest, namely for the RegMax, Born--Infeld, and Frolov--Hayward  theories.

As we shall see the Bonnor--Melvin universe in Maxwell theory is static and axisymmetric, with a regular axis of symmetry located at cylindrical $\rho=0$, and regular fields everywhere in the region $\rho\in[0,\infty)$. In what follows we shall attempt to find a generalization of such a solution to any NLE theory, keeping the above defining properties.

\subsection{Generalized Bonnor--Melvin solution}

Let us consider the cylindrical type coordinates $(t,\rho, z, \varphi)$ and the following ansatz for the metric and the vector potential:
\ba\label{Melvin}
ds^2&=&-\Omega(\rho)^2\,(dt^2-dz^2)+\frac{\rho^2}{f(\rho)}\,d{\rho^2}+f(\rho)\,d\varphi^2\,,\nonumber\\
A&=&\psi(\rho) \, d\varphi\,.
\ea

The {electromagnetic field is {then 
\be 
F=-B(\rho) d\rho \wedge d\varphi\,,
\ee 
where we have introduced the magnetic induction  $B(\rho)\equiv -\psi_{,\rho}$.} The invariants ${\cal S}$ and ${\cal P}$ now read}
\be\label{S_Melvin}
{\cal S}=\frac{B(\rho)^2}{\rho^2}\,,\quad {\cal P}=0\,.
\ee
The vanishing of ${\cal P}$  implies that for parity invariant theories (where ${\cal L}$ depends on ${\cal P}^2)$, we may effectively focus on restricted theories characterized by ${\cal L}({\cal S})\equiv{\cal L}({\cal S}, {\cal P}=0)$.

From the NLE  equations \eqref{FE} it is then possible to find the following expression for ${\cal L}_{\cal S}$:
\be \label{LS-NLE-Melvin}
{\cal L}_{\cal S}=-\frac{ H\,\rho}{\Omega(\rho)^2\,B(\rho)}\,,
\ee
where $H$ is an integration  constant. From the Einstein field equations with arbitrary model of NLE, we then find that
{
\be \label{Omega-gen}
\Omega(\rho)=\omega_1+\omega_2 \rho^2\,,
\ee 
and 
\ba \label{f-Melvin-NLE}
f&=&f_1\Omega(\rho)^2+\frac{f_2}{\Omega(\rho)}\nonumber\\
&&-\frac{4H}{3\omega_2\Omega(\rho)}\,\Bigl(\Omega(\rho)^3\int \frac{B}{\Omega(\rho)^3} \,d\rho -\int{Bd\rho}\Bigr)\,,
\ea
where $\omega_1, \omega_2, f_1$ and $f_2$ are the integration constants, and all field equations are now satisfied.
Here, nonzero $f_1$ is related to the cosmological constant (which can appear as an additive constant in the NLE Lagrangian); in what follows we set $f_1=0$.
Moreover, we have expressed the metric function $f$  in terms of yet unspecified magnetic induction $B$ -- to obtain an explicit solution for a given NLE model, we first have to solve \eqref{LS-NLE-Melvin} for this field.

In particular, for the Maxwell case characterized by ${\cal L}_{\cal S}=-\frac{1}{2}$,
one can make a specific  choice of $\omega_1$ and $\omega_2$ so that $\Omega$ takes the following form \cite{Gibbons_2001}: \ba\label{OmegaH_Gibbons}
{\Omega}=1+H^2{\rho}^2\,,
\ea
which yields
\be \label{B_Gibbons}
B=\frac{2H\rho}{\Omega^2}\,,\quad 
f=\frac{\rho^2}{\Omega^2}\,,
\ee
upon choosing $f_1=0$ and $f_2=1/H^2$\,.
In this case, the solution is regular on $\rho\in (0,\infty)$, with the axis of symmetry located at $\rho=0$.   Moreover,  the  constant $H$ governs   the magnetic field strength on the axis of symmetry, as we have 
\ba 
D&=& -2{\cal L}_{\cal S}F=-\frac{2\rho H}{\Omega(\rho)^2}\,  d\rho \wedge d\varphi\nonumber\\
&=&-\frac{2H}{\Omega(\rho)^2}e^{(\rho)}\wedge e^{(\varphi)}\equiv -D_{(\rho)(\varphi)} e^{(\rho)}\wedge e^{(\varphi)}\,,
\ea
where $e^{(\rho)}=\frac{\rho}{\sqrt{f}}d\rho, \ e^{(\varphi)}=\sqrt{f}d
\varphi$ are the tetrad frames.
Obviously, as $\rho\to 0$, the tetrad components of $D$ approach the value $2H$,  $D_{(\rho)(\varphi)}\to 2H$.
At the same time, identifying $\varphi \sim \varphi+2\pi$, the axis of symmetry is also regular.

On the hand,  as shown in \cite{Gibbons_2001}, where the specific choice \eqref{OmegaH_Gibbons}  for $\Omega$ was adopted universally for any NLE, the above regularity no longer necessarily survives in the NLE case. For example, the Born--Infeld theory yields regular Melvin-type solutions only in a certain parameter regime (see also below). Moreover, even if the regular axis exists, it is no longer necessarily located at $\rho=0$, but rather at $\rho=\rho_a$, given by the largest root 
\be 
f(\rho_a)=0\,.
\ee 
In order such axis be regular, one also has to properly identify the periodicity of $\varphi$. For example, for $f'(\rho_{a})\neq 0$ and $\rho_{a}>0$ the correct period is $\frac{4\pi\rho_{a}}{f'(\rho_{a})}$ instead of $2\pi$. At the same time,   the actual magnitude of the magnetic field on the axis is no longer exactly given by the integration constant $H$.

As obvious from the above discussion, a generalization of the Bonnor--Melvin universe for NLEs has to be discussed case by case. To give specific examples, in what follows, instead of using the  formulas \eqref{LS-NLE-Melvin}--\eqref{f-Melvin-NLE}, we employ the alternative `trick' of double Wick rotation.

\subsection{Double Wick rotation}

It is well known, e.g. \cite{Gibbons_2001, Colleaux:2025uiw}, that the Bonnor--Melvin-type universe can be mapped to the corresponding planar electrically charged solution. To see this explicitly, let us return to the ansatz \eqref{Melvin}, 
where we treat the coordinates $(t,z,\varphi)$ as dimensionless, $[t]=[z]=[\varphi]=L^0$, while we assume that $[\Omega]=[\rho]=L$ and $[f]=L^2$. 
As shown in~\eqref{Omega-gen}, for any NLE the $\Omega$ takes the following form: $\Omega=\omega_1+\omega_2\rho^2$. This means that we can introduce the length scale $L_0$ and  perform the following coordinate transformation:
\be 
d\Omega=\frac{1}{L_0}\rho\, d\rho\,, 
\ee 
so that $[\Omega]=L$, 
together with the identifications $f(\rho)=L_0^2\tilde f(\Omega)$ and $ \psi(\rho)=L_0\tilde \psi(\Omega)$. This yields the following form of the Bonnor--Melvin Universe:
\ba\label{new} 
ds^2&=&-\Omega^2(dt^2-dz^2)+\frac{d\Omega^2}{\tilde f}+\tilde f L_0^2d\varphi^2\,,\nonumber\\
A&=&\tilde \psi L_0 d\varphi\,.
\ea 
In further writing $B=\tilde \psi_{,\Omega}$, 
the Ricci scalar, Riemann square, and the electromagnetic field invariant read
\ba 
{\cal R}&=&-\frac{1}{\Omega^2}\left(\Omega^2\,{\tilde f}''+4\Omega\,{\tilde f}'+2{\tilde f}\right)\,,\nonumber\\
R_{\mu\nu\kappa\lambda}R^{\mu\nu\kappa\lambda}&=&\frac{1}{\Omega^4}\bigl(\Omega^4({\tilde f}'')^2+4\Omega^2({\tilde f}')^2+4{\tilde f}^2\bigr)\,,\nonumber\\
{\cal S}&=&B^2\,.
\ea

Performing further the following double Wick rotation:
\be \label{DWick}
L_0\varphi \to i \tau\,,\quad t\to i x\,,
\ee
treating $\tau$ now as dimension-full, $[\tau]=L$,
together with $i \tilde \psi=\phi$ and $\tilde f\to h$,
we recover 
\ba \label{planar-metric}
ds^2&=&-h d\tau^2+\frac{d\Omega^2}{h}+\Omega^2(dx^2+dz^2)\,,\nonumber\\
A&=&\phi d\tau\,.
\ea 
This is nothing else than the electrically charged planar solution. 
Thus, in principle, 
in any NLE one can construct the generalized Bonnor--Melvin-type universe by a double Wick rotation of the corresponding electrically charged planar solution. 

In this map, the horizon equation $h(r_0)=0$ maps to the corresponding equation for the axis of symmetry. Note, however, that one needs to take into account the map 
\be \label{phi-Wick}
\phi=i \tilde \psi\,,   
\ee
together with the (potentially required) `transformation' of electric to magnetic Lagrangian, e.g. \eqref{Frolov-HaywardNLE} and \eqref{Frolov-Lagrangian-magnetic}. Consequently, 
additional rotation of charges and NLE parameters may be required, and 
ranges of parameters may be modified.
{As we shall see, this may  significantly alter the behavior of the metric function $h\leftrightarrow \tilde{f}$ and consequently also the position of horizon/axis.}


\subsection{Maxwell case}

{
To illustrate the above procedure, let us first consider the planar charged solution \eqref{planar-metric} in Maxwell theory, characterized by $(\Omega>0)$\footnote{\label{Footnote1}Note that due to re-parametrization freedom $\tau\to a\tau, \Omega\to a^{-1}\Omega, x\to ax, z\to az$ (supplemented by $M\to a^{-3}M, Q\to a^{-2}Q$) of planar metric \eqref{planar-metric} with the specific $h$ given above \eqref{h-Maxwell} we can set one of the constants (preferably $M$) to arbitrary value (preserving sign) and thus only one of these constants is really an independent parameter of the solution. The electric potential in \eqref{h-Maxwell}, or rather the charge within, rescales in compatible way.}
\be \label{h-Maxwell}
h=-\frac{2M}{\Omega}+\frac{Q^2}{\Omega^2}\,,\quad \phi=-\frac{Q}{\Omega}\,.
\ee 
 To obtain real potential $\tilde \psi$ upon the Wick rotation \eqref{DWick}, we set $Q\to iQ$, giving 
\be 
\tilde f=-\frac{2M}{\Omega}-\frac{Q^2}{\Omega^2}\,,\quad \tilde \psi=-\frac{Q}{\Omega}\,.
\ee 
To avoid asymptotically negative $\tilde f$, we now consider $M$ (which is simply an integration constant) to take negative values. We display the corresponding function $\tilde f$ for $M=-1$ and $Q=1$  by a black curve in Fig.~\ref{fig:Melvin-RegMax}.

While the solution is singular at $\Omega=0$, such a singularity occurs below the axis of symmetry, located at $\tilde f(\Omega=r_0)=0$. Since $r_0>0$, we may shift $\Omega$ so that the axis of symmetry is located at the origin of a new coordinate $r\equiv\Omega-r_0$,  ${\tilde f}(r=0)=0$. This yields a constraint that can, for example, be solved for $M$, giving 
\be\label{MMaxwellMelvin}
M=-\frac{Q^2}{2r_0}\,,
\ee
and 
\be 
\tilde f=\frac{Q^2r}{r_0(r+r_0)^2}\,,\quad \tilde \psi=-\frac{Q}{r+r_0}\,,\quad B=\frac{Q}{(r+r_0)^2}\,.
\ee 
Obviously, the solution is now regular for all $r>0$, with the axis of symmetry located at $r=0$. To ensure regularity at $r=0$, we look at the last two terms in \eqref{new}, introducing the new coordinate $d\tilde \rho=dr/\sqrt{\tilde f'(r=0)r}$, that is $\tilde\rho=\sqrt{r/\tilde f'(r=0)}$, so that close to $r=0$ we have 
\ba
\frac{d\Omega^2}{\tilde f}+\tilde f L_0^2d\varphi^2&\approx& d\tilde\rho^2+\tilde f'(r=0)rL_0^2d\varphi^2\nonumber\\
&=&d\tilde\rho^2+[\tilde f'(r=0)L_0]^2\tilde \rho^2d\varphi^2\,.
\ea 
Thus, to have regular axis, we have to identify
\be\label{L0K} 
\varphi\sim \varphi+\frac{2\pi}{K}\,,\quad K=\tilde f'(r=0)L_0\,.
\ee 
In our Maxwell case, this yields
\be 
K=\frac{Q^2L_0}{r_0^3}\,.
\ee 

To summarize, by construction, we have introduced 4 parameters $\{M,Q,r_0, L_0\}$. Of these, one (e.g $L_0$) is fixed by requiring the regularity of the axis, e.g. setting $K=1$ in \eqref{L0K}, another (e.g. $M$) can be eliminated by requiring ${\tilde f}(r=0)$ giving \eqref{MMaxwellMelvin}, and one of $M$ and $Q$ can be scaled away by exploiting the scaling symmetry described in Footnote~\ref{Footnote1},
giving rise to the Bonnor--Melvin Universe in Maxwell theory characterized by a single physical parameter. Such a parameter is then, for example fixed by requiring $B(r=0)=1$. 
We display the corresponding (upgraded) function $\tilde f$ for the Maxwell model 
in Fig.~\ref{fig:Melvin-RegMax1},  Fig.~\ref{fig:Melvin-BI1}, and Fig.~\ref{fig:Melvin-Frolov1} (shown by black curves; these are then used as a reference for comparing the behavior of $\tilde f$ in other NLE models, namely the RegMax, Born--Infeld, and Frolov--Hayward cases. The corresponding (upgraded) magnetic induction is similarly displayed and compared to other NLE models in Fig.~\ref{fig:B-Melvin}. 

\subsection{RegMax case}

Consider now a similar procedure in the RegMax theory.  
The planar solution now reads (taking $Q>0$ as always), c.f. \eqref{SSS-RegMax}:
\ba
h&=&-2\alpha^2Q+\frac{4\alpha Q^{3/2}}{3\Omega}-\frac{2M}{\Omega}+4\Omega\alpha^3\sqrt{Q}\nonumber\\
&&-4\alpha^4\Omega^2\log\Bigl(1+\frac{\sqrt{Q}}{\Omega\alpha}\Bigr)\,,\nonumber\\
\phi&=&-\frac{\alpha Q}{\alpha \Omega+\sqrt{Q}}\,,
\ea 
and is regular for $\Omega>0$.
To obtain real potential and to transform the electric RegMax Lagrangian \eqref{Tay}--\eqref{s} to the required magnetic one \cite{Hale:2023dpf}:
\ba 
\tilde {\cal L}_{\mbox{\tiny RegMax}}&=&2\alpha^4\Bigl(1-3\log(1-s)+\frac{s^3+3s^2-4s-2}{2(1-s)}\Bigr)\,,\nonumber\\
s&=&-\Bigl(S/\alpha^4\Bigr)^{1/4}\,,
\ea 
we perform the following Wick rotation:
\be 
Q\to iQ\,,\quad \alpha \to -\sqrt{i}\alpha\,,
\ee 
mapping ${\cal L}_{\mbox{\tiny RegMax}}\to -\tilde{\cal L}_{\mbox{\tiny RegMax}}$. This yields the following potential and metric function:
\ba \label{f-Melvin-RegMax}
\tilde f&=&2\alpha^2Q+\frac{4\alpha Q^{3/2}}{3\Omega}-\frac{2M}{\Omega}+4\Omega\alpha^3\sqrt{Q}\nonumber\\
&&+4\alpha^4\Omega^2\log\Bigl(1-\frac{\sqrt{Q}}{\Omega\alpha}\Bigr)\,,\nonumber\\
\tilde \psi&=&-\frac{\alpha Q}{\alpha \Omega-\sqrt{Q}}\,.
\ea 
The solution is now singular at $\Omega=\Omega_s\equiv \sqrt{Q}/\alpha$ {(both curvature and magnetic field diverge there)}; 
note different signs w.r.t. the planar case in various places.  Expanding for large $\Omega$, we get 
\be 
\tilde f=-\frac{2M}{\Omega}-\frac{Q^2}{\Omega^2}+O(1/\Omega^3)\,,
\ee 
thus we again take $M<0$ to have asymptotically positive $\tilde f$. At the same time, the metric function $\tilde f$ logarithmically diverges to minus infinity as the singularity $\Omega\to \Omega_s$ is approached. This means that for any $\alpha$ (and any $M<0$ and $Q>0$), there is always $r_0>\Omega_s$ for which the axis occurs, $\tilde f(\Omega=r_0)=0.$ 
We display the behavior of $\tilde{f}$ \eqref{f-Melvin-RegMax} in Fig.~\ref{fig:Melvin-RegMax} for $M=-1, Q=1$ and several values of $\alpha$, also including the standard Bonnor–Melvin ($\alpha\to\infty$ limit) for comparison. The fact that at the curvature singularity the function $\tilde{f}$ diverges (to minus infinity) is an important characteristic distinguishing RegMax from the other two NLE models considered.

By shifting $\Omega=r+r_0$, we may again move the axis to $r=0$, and obtain a fully regular RegMax Bonnor--Melvin universe. We display the corresponding (upgraded) $\tilde f$ in Fig.~\ref{fig:Melvin-RegMax1} and the associated magnetic induction in Fig.~\ref{fig:B-Melvin}; note that for $Q=1$, $B$ for RegMax and Maxwell theories coincide. Interestingly, $B$ always coincides for these two theories, which is related to the RegMax model providing a shift-like regularization of the point charge field \cite{Hale:2023dpf}. However, this magnetic field `sits on' different geometries.

\begin{figure}
    \centering
    \includegraphics[width=\linewidth]{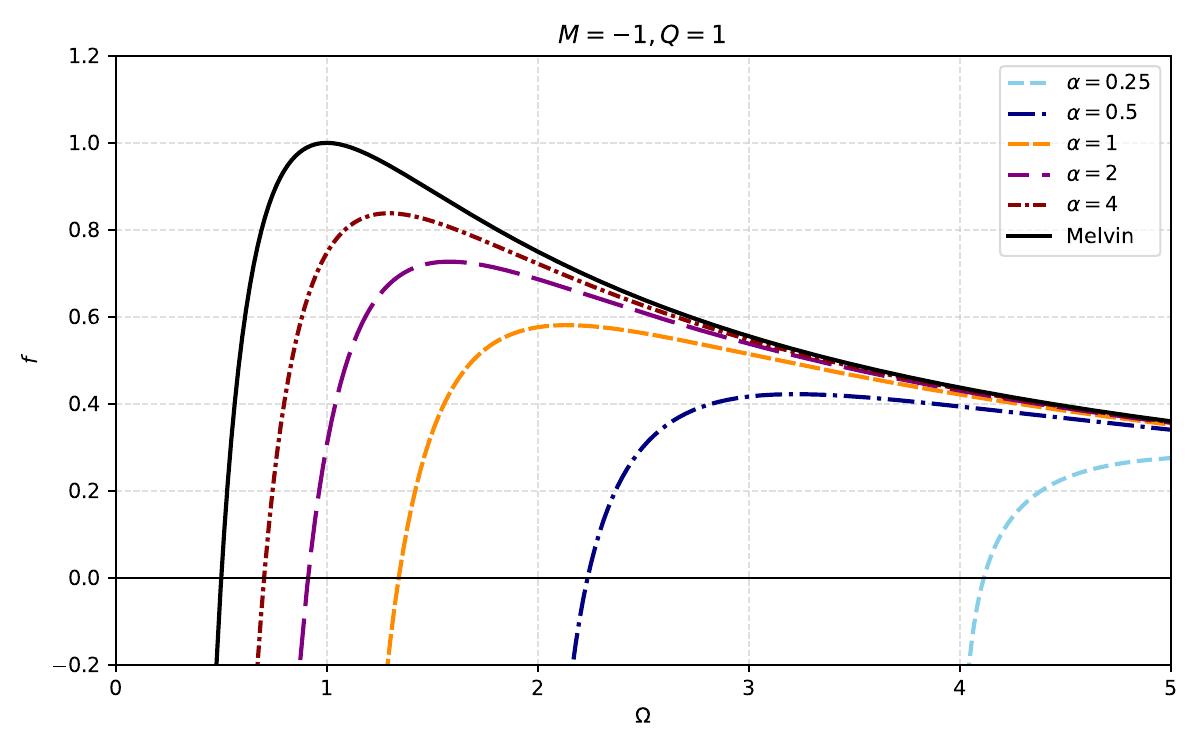}
    \caption{{\bf Metric function $\tilde f$ for RegMax model}. The metric function  $\tilde{f}$ \eqref{f-Melvin-RegMax} is displayed for $M=-1, Q=1$ and several values of the RegMax parameter $\alpha$. The standard Bonnor--Melvin Universe in Maxwell theory ($\alpha\to\infty$ limit) is displayed by a solid black curve.}
    \label{fig:Melvin-RegMax}
\end{figure}

\begin{figure}
    \centering
    \includegraphics[width=\linewidth]{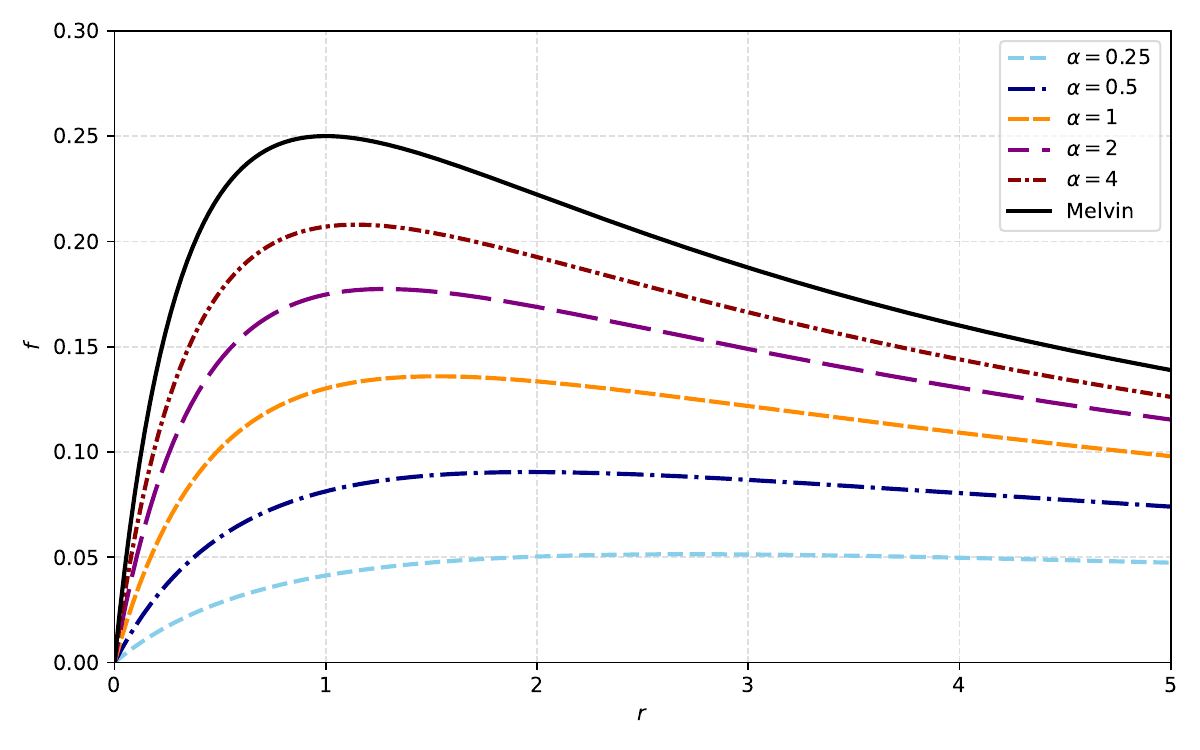}
    \caption{{\bf Bonnor--Melvin-type solution for RegMax model.} The (upgraded) metric function $\tilde f$ is displayed upon regularizing the axis and moving it to the origin of a new coordinate $r$, while demanding $B(r=0)=1$ and setting $Q=1$. The coloured curves correspond to various values of the RegMax parameter $\alpha$, with the Maxwell case displayed by a black curve. }
    \label{fig:Melvin-RegMax1}
\end{figure}

\subsection{Born--Infeld case}
Let us next turn to the Born--Infeld case. In this case, the planar solution reads (c.f. \eqref{SSS-BI} or e.g. \cite{Gunasekaran:2012dq}):
\ba 
h&=&-\frac{2M}{\Omega}+\frac{2b^2}{\Omega}\int_\Omega^\infty\Bigl( \sqrt{\Omega^4+\frac{Q^2}{b^2}}-\Omega^2\Bigr) d \Omega\nonumber\\
&=&-\frac{2M}{\Omega}
+\frac{2b^2\Omega^2}{3}\Bigl(1-\sqrt{1+\frac{Q^2}{b^2\Omega^4}}\Bigr)\nonumber\\
&&+\frac{4Q^2}{3\Omega^2}{}_2F_1\Bigl(\frac{1}{4},\frac{1}{2};\frac{5}{4};-\frac{Q^2}{b^2\Omega^4}\Bigr)\,,
\ea 
and 
\be 
E=\phi_{,\Omega}=\frac{Q}{\sqrt{\Omega^4+Q^2/b^2}}\,.
\ee 
In this case we only Wick rotate the charge $Q\to iQ$ and keep the Born--Infeld parameter as is. So we have 
\ba \label{f-Melvin-BI}
\tilde f&=&-\frac{2M}{\Omega}
+\frac{2b^2\Omega^2}{3}\Bigl(1-\sqrt{1-\frac{Q^2}{b^2\Omega^4}}\Bigr)\nonumber\\
&&-\frac{4Q^2}{3\Omega^2}{}_2F_1\Bigl(\frac{1}{4},\frac{1}{2};\frac{5}{4};\frac{Q^2}{b^2\Omega^4}\Bigr)\,,\nonumber\\
B&=&{\tilde \psi}_{,\Omega}=\frac{Q}{\sqrt{\Omega^4-Q^2/b^2}}\,.
\ea 
Again, we have to have $M<0$. There is a singularity of both curvature and magnetic field at 
\be 
\Omega=\Omega_s=\sqrt{Q/b}\,.
\ee
As first observed in \cite{Gibbons_2001}, this time, however, this singularity may be `naked' -- not hidden behind the axis of symmetry. This happens for 
\ba \label{b-critical}
b<b_c&\equiv& 18\Gamma\Bigl(\frac{3}{4}\Bigr)^{\!\!4}\frac{2\sqrt{2}\pi^{3/2}\Gamma\Bigl(\frac{3}{4}\Bigr)^2+2\Gamma\Bigl(\frac{3}{4}\Bigr)^4+\pi^3}{\Bigl[2\Gamma\Bigl(\frac{3}{4}\Bigr)^4-\pi^3\Bigr]^2}\nonumber\\
&\approx&  3.420660 \frac{M^2}{Q^3}\,.
\ea 
On the other hand, for $b>b_c$ we have an axis of symmetry and the solution can be made regular, using a similar procedure to what happens in the Maxwell case. We display the behavior of $\tilde{f}$ \eqref{f-Melvin-BI} in Fig.~\ref{fig:Melvin-BI} for $M=-1, Q=1$ and several values of $b$, also including standard Bonnor--Melvin ($b\to\infty$ limit) for comparison. Note that abrupt end of the curves corresponding to $\tilde{f}$ with finite $b$ is caused by appearance of curvature singularity, which arises even-though $\tilde{f}$ is finite there.\footnote{One might wonder why the RegMax Melvin is always regular, but Born--Infeld has singular cases, even though on the level of electrically charged black holes they have equivalent behavior. This is caused by ${\tilde f}$ being finite and potentially positive at the point of curvature divergence for Born--Infeld unlike the case of RegMax where it diverges logarithmically to minus infinity.}

\begin{figure}
    \centering
    \includegraphics[width=\linewidth]{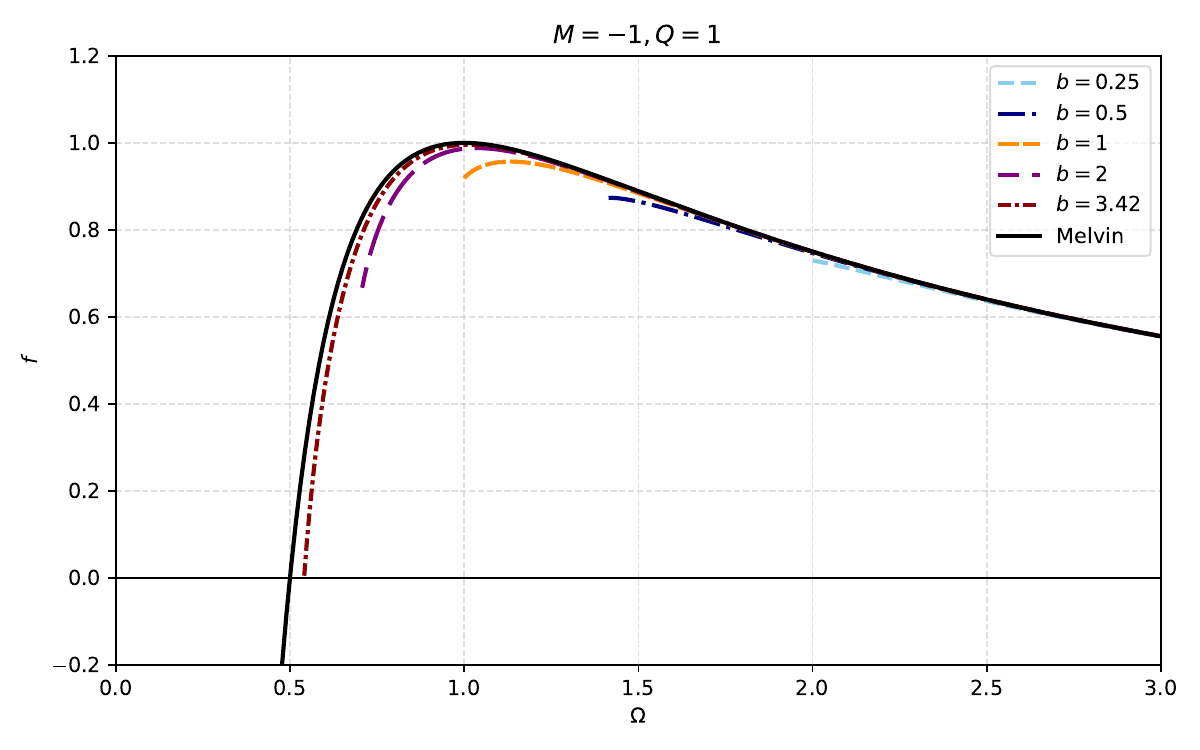}
    \caption{
    {\bf Metric function $\tilde f$ for Born--Infeld model}. The metric function  $\tilde{f}$ \eqref{f-Melvin-BI} is displayed for $M=-1, Q=1$ and several values of the Born--Infeld  parameter $b$. The standard Bonnor--Melvin Universe in Maxwell theory ($b \to\infty$ limit) is displayed by a solid black curve. Note that $b=3.42$ corresponds to approximate value of $b_{c}$ \eqref{b-critical} for selected parameters; for $b>b_c$ the axis exists. It is these $b>b_c$ solutions that give raise to regular Born--Infeld Bonnor--Melvin-type solutions displayed in the next figure. }
    \label{fig:Melvin-BI}
\end{figure}
\begin{figure}
    \centering
    \includegraphics[width=\linewidth]{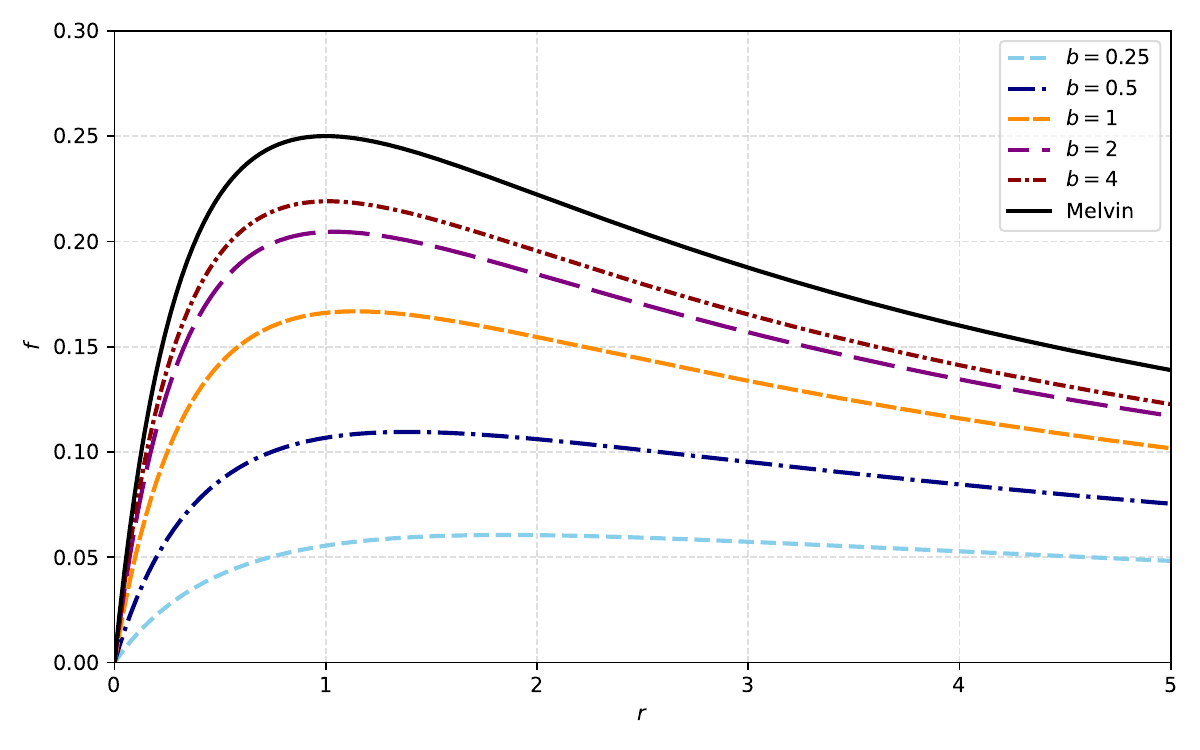}
    \caption{{\bf Bonnor--Melvin-type solution for Born--Infeld model.} The (upgraded) metric function $\tilde f$ is displayed upon regularizing the axis and moving it to the origin of a new coordinate $r$, while demanding $B(r=0)=1$ and setting $Q=1$. The coloured curves correspond to various values of the Born--Infeld parameter $b$, with the Maxwell case displayed by a black curve. {Note that by adjusting $M^2/Q^3$, the regular solution exists for any value of the parameter $b$.}}
    \label{fig:Melvin-BI1}
\end{figure}

Note, however, that the apparent non-existence of the regular axis is an `artifact' of fixing $M$ and $Q$ in \eqref{f-Melvin-BI} and \eqref{b-critical}; a regular Born--Infeld Bonnor--Melvin-type solution can be obtained for any value of the Born--Infeld parameter $b$ by adjusting the value of $M^2/Q^3$ (this ratio is invariant with respect to the re-parametrization discussed in Footnote \ref{Footnote1}, which carries over to the Wick rotated case), that is, by demanding that the axis of symmetry exists. Using this and the shift procedure as in the previous two cases, we then obtain the (upgraded) metric function $\tilde f$ characterizing the  
Born--Infeld Bonnor--Melvin-type solution displayed in Figs.~\ref{fig:Melvin-BI1} and \ref{fig:B-Melvin}.

\subsection{Frolov--Hayward  case}
 Let us finally turn to the Frolov--Hayward electrodynamics. 
 Its static spherically symmetric solution is derived in the Appendix \eqref{App}, and the corresponding planar 
solution (regular for $\Omega>0$) reads
 \ba
h&=&-\frac{2M}{\Omega}-\frac{2\beta^2Q}{3}+\frac{4\beta Q^{3/2}}{3\Omega}\arcCot\Bigl(\frac{\Omega\beta}{\sqrt{Q}}\Bigr)\nonumber\\
&&
+\frac{2\beta^4 \Omega^2}{3}\log\left(1+\frac{Q}{\beta^2\,\Omega^2}\right)\,,\nonumber\\
E&=&\phi_{,\Omega}=\frac{ Q}{\Omega^2+Q/\beta^2}\,. 
\ea
To obtain the real potential and transform the electric Lagrangian \eqref{Frolov-HaywardNLE} into the required  magnetic one, as explained in \cite{Hale:2023dpf}:
\ba
\tilde{\cal L}_{\mbox{\tiny FH}}={\beta^4}\bigl(\hat s+\ln(1-\hat s)\bigr)\,,\quad \hat s=-\sqrt{{\cal S}/\beta^4}\,,
\ea
we perform the following Wick rotation:
\be 
Q\to iQ\,,\quad \beta \to \sqrt{-i}\beta\,,
\ee 
mapping ${\cal L}_{\mbox{\tiny FH}}\to -\tilde{\cal L}_{\mbox{\tiny FH}}$. This yields the following metric function and magnetic field:
\ba \label{f-Melvin-Frolov}
\tilde{f}&=&-\frac{2M}{\Omega}-\frac{2\beta^2Q}{3}-\frac{4\beta Q^{3/2}}{3\Omega}\arcCoth\Bigl(\frac{\Omega\beta}{\sqrt{Q}}\Bigr)\nonumber\\
&&
-\frac{2\beta^4 \Omega^2}{3}\ln\left(1-\frac{Q}{\beta^2\,\Omega^2}\right)\,,\nonumber\\
B&=&{\tilde \psi}_{,\Omega}=\frac{ Q}{\Omega^2-Q/\beta^2}\,. \label{B_Frolov}
\ea
{Again, we have to consider $M<0$. The solution is singular at $\Omega=\Omega_s=\sqrt{Q/\beta^2}$ and similar to the Born--Infeld case this singularity may be `naked' -- not hidden behind the axis of symmetry. This happens for 
\be \label{beta-critical}
\beta<\beta_{c}\equiv -\frac{3M}{Q^{3/2}(1+\ln{4})}\approx  -1.25718 \frac{M}{Q^{3/2}}\,.
\ee 
On the other hand, for $\beta>\beta_c$ we have an axis of symmetry and the solution can be made regular. We display the behavior of $\tilde{f}$ \eqref{f-Melvin-Frolov} in Fig.~\ref{fig:Melvin-Frolov} for $M=-1, Q=1$ and several values of $\beta$, also including standard Bonnor--Melvin ($\beta\to\infty$ limit) for comparison. As in the Born--Infeld case, the abrupt end of the curves corresponding to $\tilde{f}$ for some $\beta$ is caused by appearance of curvature singularity for finite values of $\tilde{f}$.}

{
As in the Born--Infeld case, for any value of the parameter $\beta$, we can (by adjusting the re-parametrization invariant quantity $M/Q^{3/2}$) focus on solutions with axis of symmetry, giving rise to the regular Born--Infeld Bonnor--Melvin-type solutions displayed in Figs.~\ref{fig:Melvin-Frolov1} and \eqref{fig:B-Melvin}.  
}

\begin{figure}
    \centering
    \includegraphics[width=\linewidth]{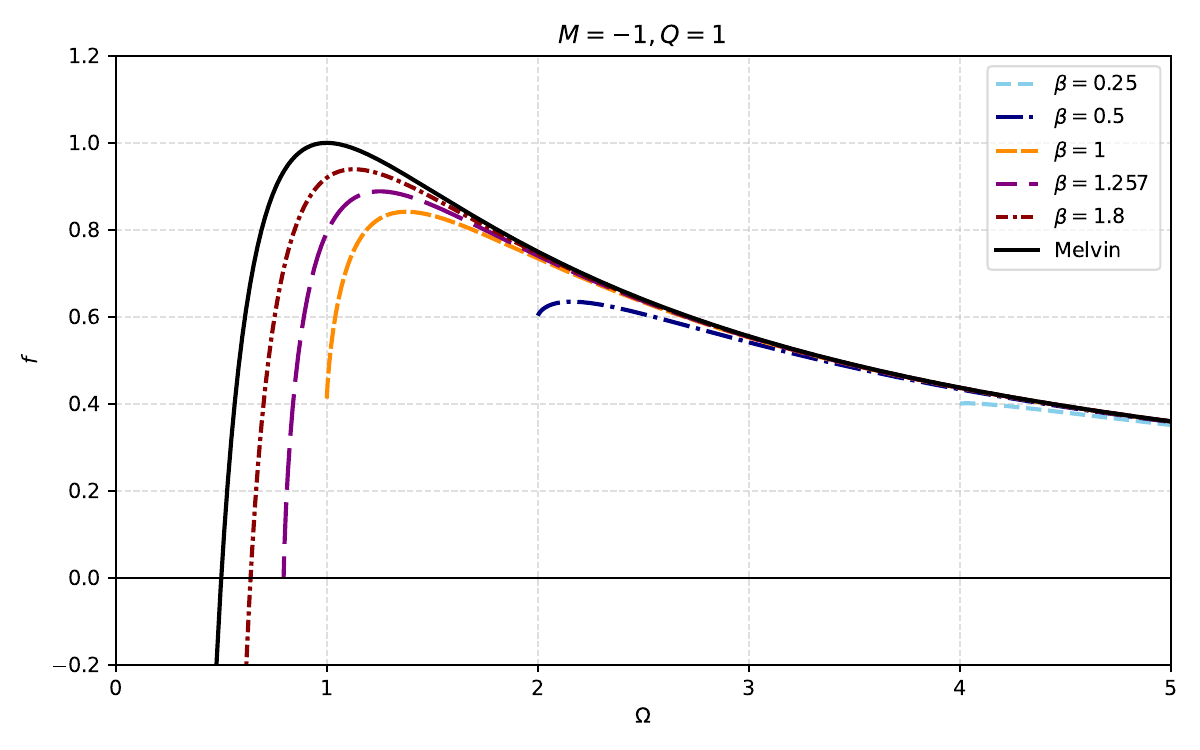}
    \caption{
    {\bf Metric function $\tilde f$ for Frolov--Hayward model}. The metric function  $\tilde{f}$ \eqref{f-Melvin-Frolov} is displayed for $M=-1, Q=1$ and several values of the Frolov--Hayward parameter $\beta$. The standard Bonnor--Melvin Universe in Maxwell theory ($\beta \to\infty$ limit) is displayed by a solid black curve. Note that $\beta=1.257$ corresponds to approximate value of $\beta_{c}$ \eqref{beta-critical} for selected parameters; for $\beta>\beta_c$ the axis exists. It is these $\beta>\beta_c$ solutions that give raise to regular Frolov--Hayward Bonnor--Melvin-type solutions displayed in the next figure. 
    }
    \label{fig:Melvin-Frolov}
\end{figure}

\begin{figure}
    \centering
    \includegraphics[width=\linewidth]{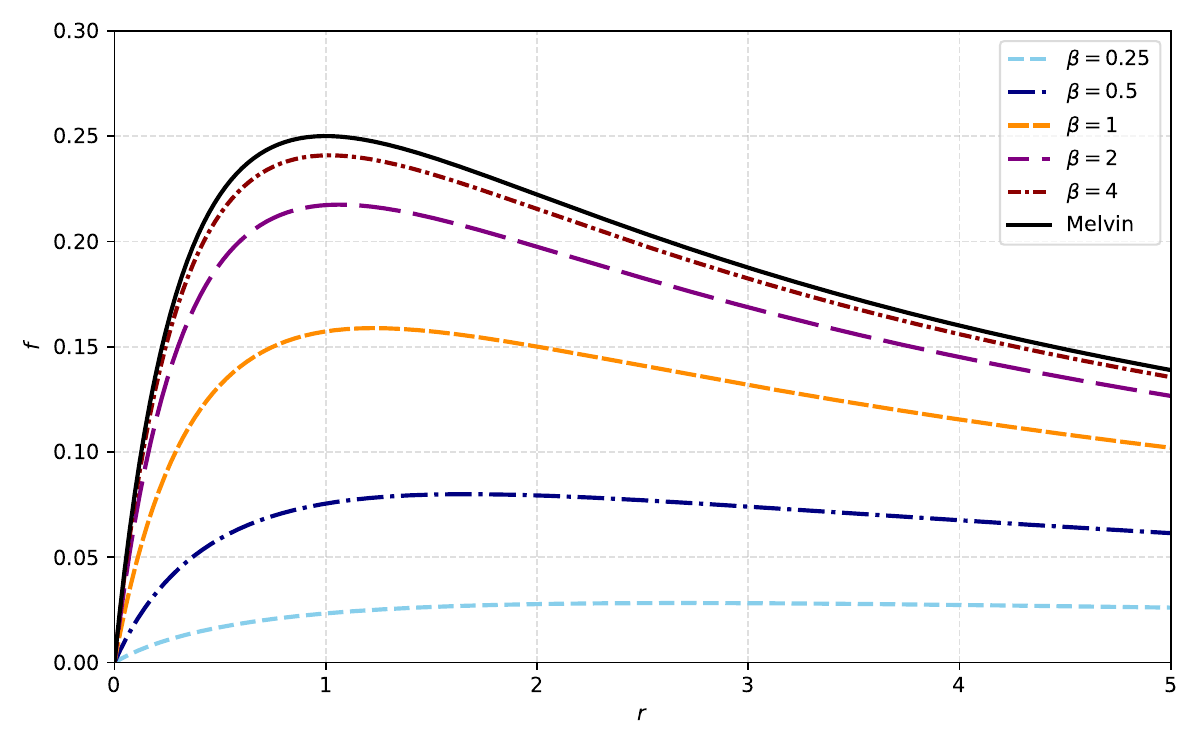}
    \caption{
    {\bf Bonnor--Melvin-type solution for  Frolov--Hayward model.} The (upgraded) metric function $\tilde f$ is displayed upon regularizing the axis and moving it to the origin of a new coordinate $r$, while demanding $B(r=0)=1$ and setting $Q=1$. The coloured curves correspond to various values of the Frolov--Hayward  parameter $\beta$, with the Maxwell case displayed by a black curve. 
    {Note that by adjusting $M/Q^{3/2}$, the regular solution exists for any value of the parameter $\beta$.}}
    \label{fig:Melvin-Frolov1}
\end{figure}

\begin{figure}
    \centering
    \includegraphics[width=\linewidth]{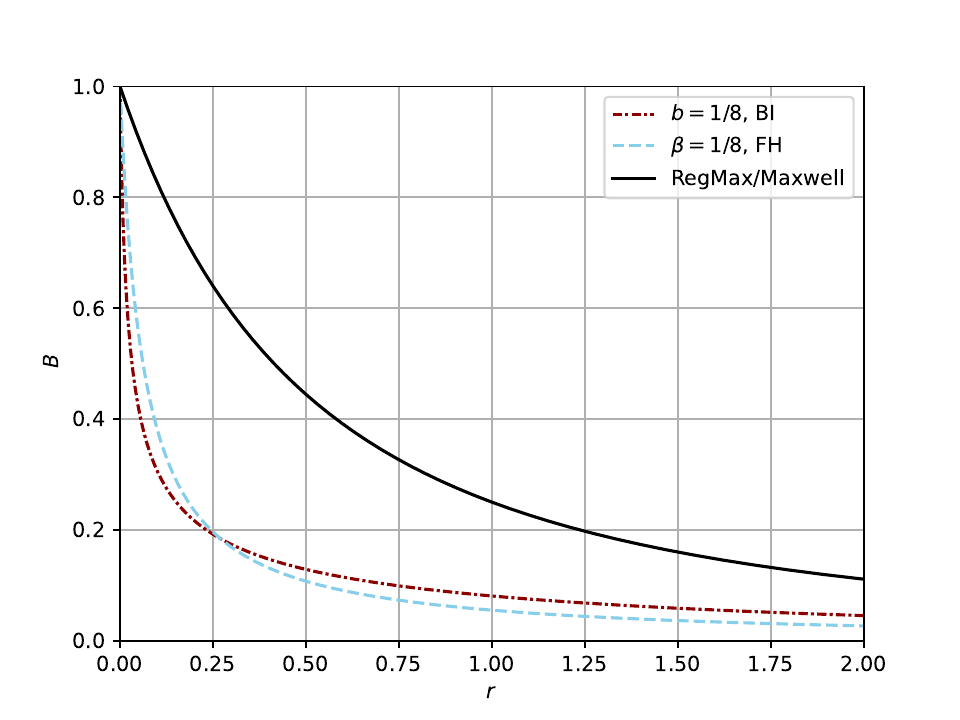}
    \caption{{\bf Magnetic induction $B$ for various NLE models}. The magnetic induction $B$ is  displayed for various NLE Bonnor--Melvin-type solutions after shifting the coordinate axis to the origin.     
    We have set $Q=1$ and demanded $B(r=0)=1$ in all cases. Note that the RegMax and Maxwell cases coincide.}
    \label{fig:B-Melvin}
\end{figure}

\section{Conclusion}\label{Sec:6}
In this paper, we have analyzed the Bertotti--Robinson type solutions in theories of NLE, showing that 
any NLE model admits a generalized Bertotti--Robinson solution of the form \eqref{final-metric-Bertotti}.
Unlike the Maxwell case, the corresponding radii of $AdS_{2}$ and $S_{2}$ geometries (which form a generalized Bertotti--Robinson solution as their direct product) are generally different, unless the corresponding NLE is conformally invariant. Such a geometry is then no longer conformally flat. Moreover, in order to solve the NLE equations for a particular model, the four constants \eqref{4-constants} are constrained by \eqref{T-Bertotti2}, reducing the number of independent parameters to two. (This hints at the potential degeneracy of geometry with respect to different NLE models.) We have provided explicit formulas for the two radii for several specific NLE models and generalized the analysis to the case involving a cosmological constant. All these results were summarized in Table \ref{Table:BR}. 

Subsequently, we have established an explicit connection between these results and the near-horizon geometries of extremal charged NLE black holes. 
{Interestingly, several of the studied models feature a charge gap for the existence of extremal black holes/ Bertotti--Robinson solutions. This can be traced to the existence of Schwarzschild-type (weakly charged) solutions without inner Cauchy horizons that are characteristic for NLE models with finite self energy \cite{Hale:2025ezt, Hale:2025urg, Russo:2026vnj}.  
We have also analyzed the linear stability of the specific NLE extremal black holes and demonstrated that they can provide particle-like models provided their interior is replaced with the corresponding generalized Bertotti--Robinson universe. 
}

In the second part of the paper, we have reviewed and expanded the analysis of Bonnor--Melvin solutions for NLE models given in \cite{Gibbons_2001}. We have shown that in the case of the RegMax model, we always obtain a regular Bonnor--Melvin solution, unlike in the cases of the Born--Infeld (the singular behavior for a certain range of parameters was already spotted in \cite{Gibbons_2001}) and Frolov--Hayward models. 
{Note, however, that the existence of such singular solutions does not prevent one from having the regular Bonnor--Melvin-type solutions for every possible value of the corresponding NLE parameter.}
We have also clarified the issues arising in the double Wick rotation approach. Namely, the necessity to also Wick rotate the NLE parameters for models whose Lagrangian has to be non-trivially generalized for magnetic solutions.

Let us end by noting that not all of the  NLE models need to admit Bonnor--Melvin-type solutions. One such exception is the ``square root'' Lagrangian \cite{Nielsen:1973qs} which does not admit  electrically charged spherically symmetric solutions and therefore does not have a planar one as well. So we cannot apply the Wick rotation trick. The same  conclusion can also be confirmed by  direct computation. It leads to an equation without solutions, confirming  that the ``square root'' model does not allow for the Bonnor--Melvin-type solutions.

\section*{Acknowledgments}

We would like to thank Robie Hennigar, Maciej Kolanowski,  and Ivan Kol{\'a}{\v r} for discussions. 
D.K. acknowledges support from the Charles University Research Center Grant No. UNCE24/SCI/016.

\appendix

\section{Frolov--Hayward black holes}\label{App}
\begin{figure}
	\begin{center}
		\includegraphics[scale=0.6]{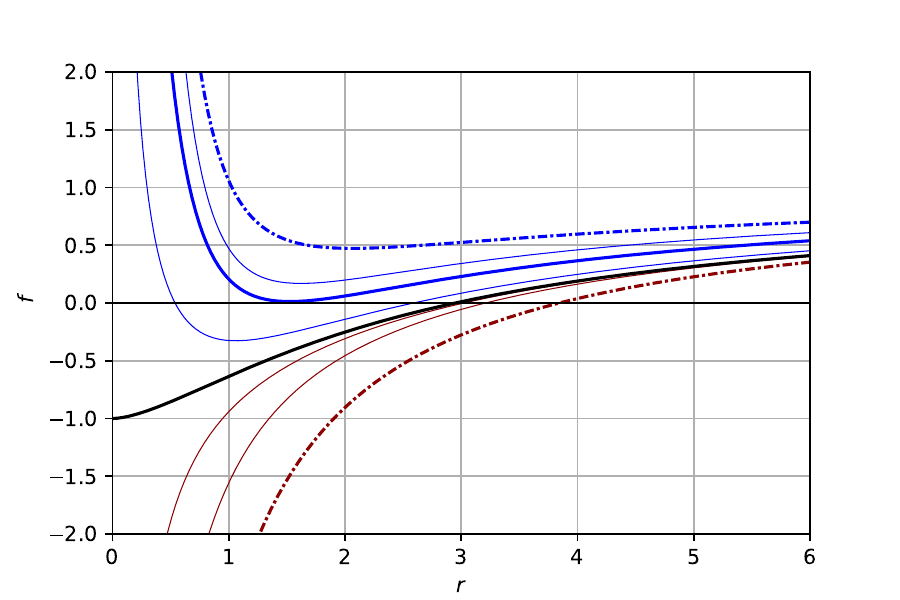}
		\caption{{\bf Frolov--Hayward black holes.} We display the three qualitatively distinct types of behavior of the metric function $f$. Namely, the blue curves correspond to the Reissner--N{\"o}rdstrom-type black holes ($M<U_{\mbox{\tiny self}}$), the red curves to the Schwarzschild-type branch black holes ($M>U_{\mbox{\tiny self}}$), and the black curve to the marginal case between the two. As is the case for any NLE with finite field strength, there exists a charge gap for the existence of Reissner--N{\"o}rdstrom-type (and thence also extremal) black holes. }\label{Fig-f-Frolov}
	\end{center}
\end{figure}
In this appendix, we present the static spherically symmetric black hole solutions for the recently obtained Frolov--Hayward  NLE model \cite{Frolov:2025ddw}. 
The corresponding Lagrangian is given by \eqref{Frolov-HaywardNLE}. It 
has a Maxwell limit upon setting $\beta\to \infty$.  

Assuming the standard spherically symmetric ansatz:
\ba 
ds^2&=&-fdt^2+\frac{dr^2}{f}+r^2d\Omega^2\,,\nonumber\\
F&=&Edr\wedge dt\,,
\ea 
where $E\equiv \psi_{,r}$, in terms of the vector potential $A=\psi dt$, we obtain:
\ba\label{f-FrolovApp}
f&=&1-\frac{2\beta^2Q}{3}-\frac{2M}{r}+\frac{4\beta Q^{3/2}}{3r}\arcCot\Bigl(\frac{r\beta}{\sqrt{Q}}\Bigr)\nonumber\\
&&
+\frac{2\beta^4 r^2}{3}\log\left(1+\frac{Q}{\beta^2\,r^2}\right)\,,\nonumber\\
E&=&\frac{Q}{r^2+Q/\beta^2}\,, \label{E_Frolov}
\ea
where $M$ represents the ADM mass and $Q$ is the asymptotic charge.

Obviously, the electric field is finite in origin and has finite self-energy.
Depending on the choice of parameters, the solution describes a black hole with one or two horizons, or a naked singularity, see Fig.~\ref{Fig-f-Frolov}.
In particular, performing the expansion of $f$ around the origin at $r=0$, we find 
\be \label{Frolov_expansion}
f=\frac{2(U_{\mbox{\tiny self}}-M)}{r}+1-2Q\beta^2+O(r^2)\,,
\ee
where $U_{\mbox{\tiny self}}$ is the point particle self energy,
\be 
U_{\mbox{\tiny self}}= \frac{\pi \beta Q^{3/2}}{3}\,.
\ee 
For large masses $M>U_{\mbox{\tiny self}}$, we have Schwarzschild-like branch with a single horizon, whereas for $M<U_{\mbox{\tiny self}}$ the behavior is Reissner--N{\"o}rdstrom like, with zero, one, or two horizons. 
As is the case for any NLE with finite self-energy of point particles \cite{Hale:2025ezt,Hale:2025urg,Russo:2026vnj}, there exists a charge (and mass) gap for the existence of the Reissner--N{\"o}rdstrom branch, determined from the next to leading term in the expansion \eqref{Frolov_expansion}, namely 
\be \label{Frolov_bound_App}
Q>Q_\star=\frac{1}{2\beta^2}\,.
\ee 
Weakly charged black holes with $Q<Q_\star$ are necessarily Schwarzschild-like -- no extremal black holes can exist in this gap. We refer to papers \cite{Hale:2025ezt,Hale:2025urg,Russo:2026vnj} for a thorough discussion on these points.

Note that, due to the fact that 
\be\label{LS-spherical-Frolov}
{\cal L}_{\cal S}=-\frac{1}{2\bigl(1-\hat s\bigr)}\,,
\ee
Frolov--Hayward model is very similar to the recently studied RegMax model, e.g. \cite{Hale:2023dpf}. For the latter, however, one can easily construct self-gravitating solutions, such as the C-metric, the Robinson--Trautman geometry, or the slowly rotating black holes. Currently, no such solutions are known for the Frolov--Hayward model.

Finally, in order to find the magnetically charged solutions in the Frolov--Hayward theory, one has to re-define the corresponding Lagrangian, as discussed in \cite{Hale:2023dpf} for the RegMax theory. Namely, one can take the Lagrangian \eqref{Frolov-HaywardNLE}, redefine $\hat s$, and flip the overall sign, to obtain
\be \label{Frolov-Lagrangian-magnetic}
{\cal L}_{\mbox{\tiny FH}}^{\mbox{\tiny (mag)}}={\beta^4}\bigl(\hat s+\ln(1-\hat s)\bigr)\,,\quad 
\hat s=-\sqrt{{\cal S}/\beta^4}\,.
\ee
With this, the magnetically charged solution for $f$ would be still given by \eqref{f-FrolovApp}, with the electric charge $Q$ simply replaced by the magnetic charge $Q_m$. It is this  Lagrangian \eqref{Frolov-Lagrangian-magnetic}, which is used to construct the Frolov--Hayward Melvin universe in the next appendix and in the main text.


\input{BR_ArXiv.bbl}

\end{document}

%% file: BR_ArXiv.bbl
%